\documentclass[sn-basic,iicol]{sn-jnl}

\usepackage{multirow}%
\usepackage{amsmath,amssymb,amsfonts}%
\usepackage{amsthm}%
\usepackage{mathrsfs}%
\usepackage{textcomp}%

\usepackage[ruled]{algorithm2e}
\usepackage[title]{appendix}
\usepackage{cleveref}
\usepackage{enumerate}
\usepackage{mathabx}
\usepackage{thm-restate}
\usepackage{xcolor}
\usepackage{xfrac}
\usepackage{multirow}
\usepackage{mathtools}
\usepackage{dblfloatfix}
\usepackage{booktabs, caption, longtable, colortbl, array}

\SetKwInOut{Input}{Input}
\SetKwInOut{Output}{Output}

\Crefname{algocf}{Algorithm}{Algorithms}

\newcommand{\paref}[1]{(\ref{#1})}

\newcommand{\review}[1]{\textcolor{black}{#1}}

\newcommand{\Cov}{{\mathbb C}\text{ov}}
\newcommand{\mE}{\mathbb E}
\newcommand{\mEt}{\widetilde{\mE}}
\newcommand{\mEtk}{\widetilde{\mE}^{(k)}}
\newcommand{\mEtun}{\widetilde{\mE}^{^{(1)}}}
\newcommand{\mEtdeux}{\widetilde{\mE}^{^{(2)}}}
\newcommand{\Jun}{J^{(1)}}
\newcommand{\tJun}{\widetilde J^{(1)}}
\newcommand{\Jdeux}{J^{(2)}}
\newcommand{\tJdeux}{\widetilde J^{(2)}}
\newcommand{\pt}{\widetilde{p}}
\newcommand{\pta}{p_{\theta}}
\newcommand{\ptun}{\widetilde{p}^{(1)}}
\newcommand{\ptdeux}{\widetilde{p}^{(2)}}
\newcommand{\Pcal}{\mathcal{P}}
\newcommand{\ta}{\theta}
\newcommand{\Var}{\mathbb V}

\newcommand{\trace}{\operatorname{Tr}}
\newcommand{\matr}[1]{\boldsymbol{#1}}

\newcommand{\sumzero}{\sum_{i_0, j_0}}
\newcommand{\sumun}{\sum_{i_1, j_1}}
\newcommand{\tr}{^{\top}}
\newcommand{\logdet}[1]{\log\left|#1\right|}
\DeclareMathOperator*{\argmax}{arg\,max}
\DeclareMathOperator{\logit}{logit}

\DeclareMathOperator{\Diag}{Diag} 
\DeclareMathOperator{\diag}{diag} 

\theoremstyle{thmstyleone}%
\newtheorem{theorem}{Theorem}

\newtheorem{proposition}[theorem]{Proposition}%

\theoremstyle{thmstyletwo}%
\newtheorem{remark}{Remark}%

\theoremstyle{thmstylethree}%

\begin{document}

\title[Zero-inflation in the MPLN Family]{Zero-inflation in the Multivariate Poisson Lognormal Family}


\author[1]{\fnm{Bastien}~\sur{Batardière}\orcid{0009-0001-3960-7120}}\email{bastien.batardiere@inrae.fr}
\author*[1]{\fnm{Julien}~\sur{Chiquet}\orcid{0000-0002-3629-3429}}\email{julien.chiquet@inrae.fr}
\author[2]{\fnm{François}~\sur{Gindraud}\orcid{0000-0002-9000-5707}}\email{francois.gindraud@inrae.fr}
\author[3]{\fnm{Mahendra}~\sur{Mariadassou}\orcid{0000-0003-2986-354X}}\email{mahendra.mariadassou@inrae.fr}

\affil*[1]{Université Paris-Saclay, AgroParisTech, INRAE, UMR MIA Paris-Saclay, 91120, Palaiseau, France.\orgdiv{Department}, \orgname{Université Paris-Saclay, AgroParisTech, INRAE, UMR MIA Paris-Saclay}, \postcode{91120}, \city{Palaiseau}, \country{France}}

\affil[2]{\orgname{Université Lyon 1, CNRS, Laboratoire de Biométrie et Biologie Evolutive UMR 5558}, \postcode{F-69622}, \city{Villeurbanne},  \country{France}}

\affil[3]{\orgname{Université Paris-Saclay, INRAE, MaIAGE}, \city{City}, \postcode{78350}, \city{Jouy-en-Josas}, \country{France}}

\abstract{Analyzing high-dimensional count data is a challenge and statistical model-based approaches provide an adequate and efficient framework that preserves explainability. The (multivariate) Poisson-Log-Normal (PLN) model
is one such model: it assumes count data are driven by an underlying structured latent Gaussian variable, so that the dependencies between counts solely stems from the latent dependencies. However PLN doesn't account for zero-inflation, a feature frequently observed in real-world datasets. Here we introduce the Zero-Inflated PLN (ZIPLN) model, adding a multivariate zero-inflated component to the model, as an additional Bernoulli latent variable. The Zero-Inflation can be fixed, site-specific, feature-specific or depends on covariates. We estimate model parameters using variational inference that scales up to datasets with a few thousands variables and compare two approximations: (i) independent Gaussian and Bernoulli variational distributions or (ii) Gaussian variational distribution conditioned on the Bernoulli one. The method is assessed on synthetic data and the efficiency of ZIPLN is established even when zero-inflation concerns up to $90\%$ of the observed counts. We then apply both ZIPLN and PLN to a cow microbiome dataset, containing $90.6\%$ of zeroes. Accounting for zero-inflation significantly increases log-likelihood and reduces dispersion in the latent space, thus leading to improved group discrimination.}




\maketitle

\section{Introduction}
\label{sec:intro}


\review{
Count data are ubiquitous in fields such as ecology, accident analysis, 
single-cell RNA sequencing, and metagenomics. A common analytical goal is to 
estimate correlations between species abundances in a biome or gene expressions 
in a tissue. This work is motivated by microbiome studies, where a marker gene 
is sequenced using high throughput technologies and the sequences are processed 
to produce Operational Taxonomic Units (OTUs) or Amplicon Sequence Variants 
(ASVs), which serve as proxies for microbial species. Microbiome studies 
generate count tables characterized by a high fraction of zero counts (80–95\%) 
and a large number of variables.
}

\review{
Analyzing count data directly is challenging, and 
transformations are often required to extract meaningful statistics. While 
log-transformation followed by Gaussian analyses is widely used, it lacks 
statistical rigor \citep{NoLogTransform}. Model-based approaches, such as 
Negative-Binomial (NB) and Poisson models, are more appropriate and have been 
extensively applied in RNAseq studies \citep[see \emph{e.g.}][]{DEseq2}. The NB 
distribution, which combines a Poisson emission law with a Gamma-distributed 
parameter, is preferred to the standard Poisson for modeling overdispersion—a 
common feature in sequencing-based count data \citep[including scRNA-seq, 
see][]{overdispersion}. However, neither the NB nor the Poisson distribution can 
model the fraction of zero counts independently from the mean count.
}

\review{
Univariate Zero-Inflated Poisson (ZIP) models \citep{lambert1992} address zero 
inflation by incorporating a Dirac mass at zero, enabling efficient differential 
analyses but disregarding inter-variable dependencies. Extensions to 
Multivariate ZIP (MZIP) and Zero-Inflated Negative-Binomial (ZINB) models have 
been proposed, but their correlation structures remain constrained \citep{MZIP,MZINB}. A bivariate ZINB model has been used to measure  correlation between two genes in scRNAseq data but does not scale 
to higher order dependencies \citep{BZINB}.
}


\review{
The (multivariate) Poisson-Log-Normal \citep[in short PLN, see][]{Aitchison1989} 
model offers a flexible framework for multivariate count data, capturing 
dependencies between counts through latent Gaussian variable.
As a mixture of Poisson with Log-Normal distributed parameters, PLN models 
naturally results in overdispersion and, unlike the NB, correlation 
between variables. The PLN model falls in the family of latent variable models 
(LVMs), and more specifically of multivariate generalized linear mixed models
(mGLMMs) sometimes also called generalized linear latent variable models 
(GLLVMs). Parameter estimation for common GLLVMs is efficiently implemented in 
some packages \citep{gllvm,statsmodels}, making it 
a popular option for multivariate count data. However, while some models allows 
for dependency between variables or zero-inflation, to the best of our 
knowledge, only one \citep{Lee2020} accounts for both at the same time.
}


\review{
We consider here the Zero-Inflated Poisson Log-Normal (ZIPLN) model, based on
the PLN model and also used in \citet{Lee2020}. ZIPLN benefits from the Gaussian 
structure of the PLN model, with an extra zero-inflated component. This extra 
layer depends on parameters that can be shared across individuals, features, or 
linked to covariates. Exact inference of (ZI)PLN using the 
Expectation-Maximization (EM) algorithm \citep{DLR77} is intractable and we 
resort instead to variational inference 
\citep{JaJ00,WAJ08,hui2017variational,blei2017variational}. Other approaches 
based on Monte Carlo
techniques \citep{JACQUIER2007615,Cappe_mcmc,stoehr2024composite,Lee2020} have 
been proposed to provide estimators, but they do no scale
well with the dimension of the observations and cannot address medium to high 
dimension problems. Numerical integration can be performed \citep{Aitchison1989} 
as an alternative to the variational approximation used here but becomes
prohibitive when the number of dimensions exceeds $5$. We introduce in this 
work two Variational-EM algorithms based on two different approximations. The 
first assumes conditional independence between both 
components, resulting in a fast M step. By contrast, the second is slightly 
slower but leverages the dependence between components to use a more complex 
variational approximation.
}

\paragraph*{Related works}

\review{
Our work is related to ZINBWaVE, which models zero inflation and counts using 
latent factors, assumed to capture (unobserved) sample-level covariates, but 
lacks identifiability and a clear dependency structure \citep{risso2018} as 
it's mainly interested in estimating the probability that a null count arises 
from zero-inflation. The closest model is the ZIPLN of \citet{Lee2020}, which 
uses a Bayesian MCMC framework. Our work differs from theirs in two ways: the 
modeling of the zero inflation component (logit versus probit) and more 
importantly parameter estimation (fully Bayesian framework versus variational 
approximation). The Bayesian estimation relies on MCMC techniques which become 
quite slow for medium to high dimensional problems (several minutes for p $\sim$ 
50 and n $\sim$ 250). A critical advantage of variational inference is its 
scalability to large problems (a couple of seconds for similarly sized 
problems), which is our main motivation for developing and implementing a fully 
variational approach. The trade-off is access to a measure of uncertainty on the 
estimation of parameters, which is more direct in the MCMC case, even though 
alternatives do exists for a variational approach \citep{SandwichBatardiere}.
}

Recently, deep neural networks have been proposed to model count
data. Variational autoencoders (VAE) \citep{kingma2022autoencoding} are
particularly efficient, performing dimension reduction via a latent variable
framework. \cite{VAE_count} proposes a VAE to
model high dimensional overdispersed count data based on the NB distribution
and \cite{Jin_2020} models sparse and imbalanced count data with VAE. Many VAE
models have been proposed for the sole purpose of studying scRNA-seq data (see
\emph{e.g.} \citet{Lopez2018DeepGM,Choi2021.09.15.460498,XU2023100382,WANG2018320}). Although model-based and highly effective for predictions, VAE remain significantly harder to interpret in terms of coefficients, outputs and impact of structuring factors of interest than their statistical latent variable models counterparts.
\\

Our paper is organized as follows: in \Cref{sec:model}, we introduce the PLN
model followed by the ZIPLN model. In \Cref{sec:var-prox}, we discuss the
variational inference and choices made on the variational families. In
\Cref{sec:inference}, we discuss the optimization strategy. We study the model
performances on simulated data in \Cref{sec:sim-study}, and apply it to a cow microbiome dataset in \Cref{sec:application}. We conclude in \Cref{sec:discussion} with
discussions and possible improvements.

\section{Model}  \label{sec:model}

\paragraph*{Background: Multivariate Poisson lognormal-model} The multivariate
Poisson lognormal model relates a $p$-dimensional observation count vector
$\matr{Y}_i = (Y_{i1}, \dots, Y_{ip})\in\mathbb{N}^p$ to a $p$-dimensional
vector of Gaussian latent variables $\matr{Z}_i\in\mathbb{N}^p$ with precision
matrix $\matr{\Omega}$ (that is, covariance matrix $\matr{\Sigma}
\triangleq \matr{\Omega}^{-1}$). We adopt a formulation of PLN close to a
multivariate generalized linear model, where the main effect is due to a linear
combination of $d$ covariates $\matr{x}_i\in\mathbb{R}^d$ (including an intercept).  We also let the possibility to add some offsets for the $p$
variables in each sample, that is $\boldsymbol{o}_i\in\mathbb{R}^p$:
\begin{equation}
 \label{eq:PLN-model}
    \begin{array}{rcl}
  \text{latent space } &   \matr{Z}_i \sim  \mathcal{N}( \matr{x}_i^\top\matr{B},\matr{\Omega}^{-1})\\
  \text{obs. space } &   \boldsymbol{Y}_i | \boldsymbol{Z}_i\sim\mathcal{P}\left(\exp\{ \mathbf o_i + \boldsymbol{Z}_i\}\right).
  \end{array}
\end{equation}
The $d\times p$ matrix $\matr{B}$ is the latent matrix of regression parameters. The latent covariance matrix $\matr{\Sigma}$ describes the underlying residual structure of dependence between the $p$ variables, once the covariates are accounted for. We denote by $\matr{Y}, \matr{O}, \matr{X}$ the observed matrices with respective sizes $n\times p, n\times p$ and $n\times d$ obtained by stacking row-wise the vectors of counts, offsets and covariates (respectively $\matr{Y}_i,\matr{o}_i$  and $\matr{x}_i$ ). We also denote by $\matr{Z}$ the $n\times p$ matrix of unobserved latent Gaussian vectors $\matr{Z}_i$.

\paragraph*{Zero-inflated PLN regression model} We now aim to model an excess
of zeroes in the data by adding zero-inflation to the standard PLN model
\paref{eq:PLN-model}, so that the zeroes in $\boldsymbol{Y}_i$ arise from two
different  sources: either from a component where zero is the only possible
value, or from a standard PLN component like in \Cref{eq:PLN-model}. This
two-component mixture is defined thanks to an additional latent vector
$\boldsymbol{W}_{i} = (W_{i1,\dots,ip})\in\mathbb{R}^p$ of Bernoulli random
variables, parametrized by probabilities $\matr{\pi}_i = (\pi_{i1}, \dots,
\pi_{ip})$ describing the probability that variable $j$ in sample $i$ belongs
to the pure zero component:
\begin{equation}
 \label{eq:ZIPLN-model}
    \begin{array}{rl}
  \boldsymbol{Z}_i  = (Z_{ij})_{j=1\dots p} & \sim  \mathcal{N}(\matr{x}_i^\top \matr{B}, \matr{\Omega}^{-1}), \\[1.5ex]
  \boldsymbol{W}_{i} = (W_{ij})_{j=1\dots p} & \sim \otimes_{j=1}^p \mathcal B(\pi_{ij}),   \\[1.5ex]
  Y_{ij} \, | \, W_{ij}, Z_{ij} \sim W_{ij}\delta_0 & + (1-W_{ij})\Pcal\left(e^{o_{ij} + Z_{ij}}\right),\\
  \end{array}
\end{equation}
where $\delta_0$ is the Dirac distribution and we note $\matr{\pi}$ the matrix obtained by stacking the
vectors $\matr{\pi}_1^\top, \dots, \matr{\pi}_n^\top$. Our model is flexible enough to
accommodate different parametrizations for $\pi_{ij}$ based on the availability
of covariates and/or modeling choices made by the user. We consider three
variants: non-dependent (ND), column-wise dependence (CD) and row-wise dependence (RD)
\begin{subequations}
\begin{align}
    \pi_{ij} & = \pi \in [0,1] & \small{\text{(ND)}}
\label{eq:zi-models-nd} \\
\pi_{ij} & = \logit^{-1}( \boldsymbol X^{0}\boldsymbol B^0)_{ij}, &  \small{\text{(CD)}}
\label{eq:zi-models-cd}\\
\pi_{ij} & = \logit^{-1}(\widebar{\boldsymbol{B}}^0 \widebar{\boldsymbol{X}}^{0})_{ij}, & \small{\text{(RD)}}
\label{eq:zi-models-rd}
\end{align}
\end{subequations}
where $\logit^{-1}(\cdot)$ is the logistic (or inverse logit) function, $d_0 \geq 1$,
$\boldsymbol{B}^0\in\mathbb R^{d_0\times p}$ (resp. $\bar{\boldsymbol{B}}^0\in \mathbb R^{n \times d_0}$) are regression coefficients
associated with row-wise matrix of covariates $\boldsymbol {X}^0\in \mathbb R^{n \times d_0}$ (resp.
column-wise covariates $\bar{\boldsymbol{X}}^0\in \mathbb R^{d_0\times p}$), obtained by stacking the
vectors $(\matr{x}^0_1)^\top, \dots, (\matr{x}^0_n)^\top$, which may or
may not be the same as in $\matr{X}$, the matrix of covariates in the PLN
component. The ND inflation setting allows the zero-inflation component to be
shared along all individuals and variables, a simple but slightly unrealistic
assumption, while RD (resp. CD) allows zero-inflation to be shared along all
variables (resp. individuals) as it depends only on the individual (resp.
variable) covariates. CD is useful when some variables are prone to
zero-inflation across individuals (\emph{e.g.} taxa for which the marker gene
fails to amplify) whereas RD is useful when the set of
zero-inflated variables rather depends on the individual's characteristics
(\emph{e.g.} acidophile taxa in soils with high pH). Note that when the column-wise (resp. row-wise)
covariates reduce to a vector of one $\boldsymbol{1}_n$ (resp.
$\boldsymbol{1}_p^\top$), the model corresponds to an inflation towards zero
shared across rows (resp. columns) with vector of parameters $\matr{\pi} =
(\boldsymbol{\pi}_j)\in [0,1]^p$ (resp. $\matr{\pi} = (\boldsymbol{\pi}_i)\in
[0,1]^n $). We refer to Model \ref{eq:ZIPLN-model} as the ZIPLN regression
model.

Using standard results on Poisson and Gaussian distribution, we easily derive the expectation and variance of the ZIPLN regression model. Letting $A_{ij} \triangleq \exp \left( o_{ij} + \mu_{ij} + \sigma_{jj}/2\right)$ with $\mu_{ij} = \matr{x}_i^\intercal \matr{B}_j$, then
\begin{align*} \label{eq:moments_zipln}
\mE\left(Y_{ij}\right) & =  (1-\pi_{ij}) A_{ij}\geq  0,\\
\Var\left(Y_{ij}\right) & = \mE\left(Y_{ij}\right) + (1-\pi_{ij})A_{ij}^2 \left( e^{\sigma_{jj}} - (1-\pi_{ij}) \right).
\end{align*}

In the following, we are interested in inferring the vector of parameters
$\theta$  where $\theta = \left(\matr{\Omega}, \matr{B}, \matr{\pi}\right)$, $\theta
=\left(\matr{\Omega}, \matr{B}, \matr B ^0\right) $ and $\theta =
\left(\matr{\Omega}, \matr{B}, \bar{\matr{B}} ^0\right)$ for Models
\ref{eq:zi-models-nd}, \ref{eq:zi-models-cd} and \ref{eq:zi-models-rd}
respectively, and where  $ \matr \Omega \in \mathbb S^{++}, \matr B \in \mathbb R^{d\times
p}, \matr{\pi} \in \left[0,1\right], \matr B^0 \in \mathbb R^{d_0\times p}$ and
$\bar{\matr{B}}^0 \in \mathbb R^{n\times d_0}$ with $\mathbb{S}^{++}$
the set of $p\times p$ positive-definite matrices. We first show that Model \ref{eq:ZIPLN-model} is identifiable.

\paragraph*{Identifiability of ZIPLN models} Identifiability results are
available for the ZI Poisson model \citep{Li2012} and can be generalized to
the ZIPLN regression model. To this
end, we first consider the simple ZIPLN model, (\emph{i.e.} a ZIPLN model
without covariate),  with a single sample,
in order to drop the index $i$:
\begin{equation}
    \label{eq:ZIPLN-standard}
\begin{aligned}
\matr{W} = (W_j)_{j=1\dots p} & \sim \mathcal{B}^{\otimes}(\matr{\pi}) = (\pi_1) \otimes \dots \otimes(\pi_p) \\
\matr{Z} = (Z_j)_{j=1\dots p} & \sim \mathcal{N}_p(\matr{\mu}, \matr{\Omega}^{-1})  \\
Y_j | W_j, Z_j & \sim W_j \delta_0 + (1-W_j)\mathcal{P}(e^{Z_j}), \\ Y_j \perp Y_k  | \matr{W}, \matr{Z}
\end{aligned}
\end{equation}

\begin{restatable}{proposition}{identifiability} \label{prop:standard-identifiability}
The simple ZIPLN model defined in \paref{eq:ZIPLN-standard} with parameter
$\matr{\theta} = (\matr{\Omega}, \matr{\mu}, \matr{\pi})$ and parameter space
$\mathbb{S}_p^{++} \times  \mathbb{R}^p \times (0,1)^{p} $ is identifiable.
\end{restatable}

The proof relies on the method of moments and is postponed to the appendix. We now use this result to prove identifiability of the ZIPLN regression Model \ref{eq:zi-models-cd} (proof for \ref{eq:zi-models-rd} is achieved by replacing $\matr{X}^0\matr{B}$ with $\widebar{\boldsymbol{B}}^0 \widebar{\boldsymbol{X}}$).

\begin{proposition} \label{prop:regression-identifiability}
The ZIPLN regression Model \ref{eq:ZIPLN-model} with zero-inflation defined
as in \Cref{eq:zi-models-cd} and parameter $\matr{\theta} = (\matr{\Omega},
\matr{B}, \matr{B}^0)$ and parameter space $\mathbb{S}^{++} \times
\mathcal{M}_{p,d}(\mathbb{R}) \times \mathcal{M}_{p,d}(\mathbb{R})$ is identifiable  if and
only if both $n\times d$ and $n \times d_0$ matrices of covariates $\matr{X}$ and $\matr X^0$ are full rank.
\end{proposition}

\begin{proof}
    If $\matr{X}$ (resp. $\matr{X^0}$) is not full rank, there exists $\matr{B} \neq \matr{B}'$ such
that $\matr{X}\matr{B} = \matr{X}\matr{B}'$ (resp. $\matr{X}^0\matr{B} = \matr{X}^0\matr{B}'$ ) and therefore the map
$\matr{\theta} \mapsto p_{\matr{\theta}}$ is not one-to-one. We know from
Proposition~\ref{prop:standard-identifiability} that $\matr{O}
+\matr{X}\matr{B}$ and $\matr{\pi} = \logit^{-1}(\matr{X}^0\matr{B}^0)$ are
identifiable. Since the affine function and the $\logit^{-1}$ are both one-to-one,
parameters $\matr{B}$ and $\matr{B}^0$ are identifiable as soon as the maps
$\matr{B} \mapsto \matr{X}\matr{B}$ and $\matr{B}^0 \mapsto
\matr{X}^0\matr{B}^0$ are injective, which is the case as soon as $\matr{X}$ and $\matr{X}^0$
are full rank.
\end{proof}

\section{Estimation by Variational Inference}  \label{sec:var-prox}

Our goal is to maximize the marginal likelihood. In the framework of latent models, a standard approach (e.g. with \textit{Expectation-Maximization} algorithms) uses the following decomposition by integrating over the latent variables $\matr{W}, \matr{Z}$
\begin{multline}
  \label{eq:loglik}
\log p_\theta(\matr{Y}) =  \log \frac{p_\theta(\matr{Z}, \matr{W}, \matr{Y})}{p_\theta(\matr{Z}, \matr{W} | \matr{Y})} \\ = \int_{\matr{W},\matr{Z}} \log \frac{p_\theta(\matr{Z}, \matr{W}, \matr{Y})}{p_\theta(\matr{Z}, \matr{W} | \matr{Y})} p_\theta(\matr{Z}, \matr{W}|\matr{Y}) d\matr{W}d\matr{Z}.
\end{multline}
However, for the ZIPLN model, it is untractable since the conditional
distribution $p_\theta(\matr{Z}, \matr{W} | \matr{Y})$ has no closed-form. To
overcome this issue, we rely on a variational approximation of this
distribution which will yield a lower bound of $\log
p_\theta(\cdot)$: for observation $i$, we denote by $\pt_{\psi}(\matr{Z}_i,
\matr{W}_i)$ the approximation of $p_\theta(\matr{Z}_i, \matr{W}_i
|\matr{Y}_i)$ where $\psi$ is a set of variational parameters to be optimized.
Subtracting to the untractable Expression \paref{eq:loglik} of the
log-likelihood the positive (and also untractable) quantity (known as the
Kullback-Leibler divergence)
\begin{multline*}
 KL( \pt_\psi(.) \|p_\theta(. | \matr{Y}) ) = \\ \int_{\matr{W},\matr{Z}} \log \frac{\pt_\psi(\matr{Z}, \matr{W})}{p_\theta(\matr{Z}, \matr{W} | \matr{Y})} \pt_\psi(\matr{Z}, \matr{W}) d\matr{W}d\matr{Z}
\end{multline*}
results after some rearrangements in the following Evidence Lower BOund (ELBO):
\begin{align}
J(\theta, \psi) & =  \log p_\theta(\matr{Y}) - KL(\pt_\psi(.)\| p_\theta(. | \matr{Y})) \nonumber \\
& = \int_{\matr{W},\matr{Z}} \log \frac{p_\theta(\matr{Z}, \matr{W}, \matr{Y})}{\pt_\psi(\matr{Z}, \matr{W})} \pt_\psi(\matr{Z}, \matr{W}) d\matr{W}d\matr{Z}\nonumber  \\
& = \mEt [\log p_\theta(\matr{Z}, \matr{W}, \matr{Y})] - \mEt [ \log \pt_\psi(\matr{Z}, \matr{W})]. \label{ELBO_theorique}
\end{align}
This also looks like a plugin of integral \Cref{eq:loglik}  with
$p_\theta(\matr{Z}, \matr{W} | \matr{Y})$ replaced with $\pt_\psi(\matr{Z},
\matr{W})$. An appropriate choice of variational approximation will make the integral
calculation tractable, while leading to an acceptable approximation of the
log-likelihood \citep{blei2017variational}. The choice of the variational
family is crucial as an inappropriate or too simplistic  family can lead to
bias and inconsistency in the resulting estimator \citep{sandwich}
whereas a too complex family would lead to an untractable optimization criterion.

\subsection{Choice of the variational family}

\paragraph*{Standard variational approximation} A straightforward, yet efficient,
approach is to consider the mean field approximation, which breaks
all dependencies between the vectors $\matr Z_i$ and $\matr W_i$ and their
respective coordinates and approximates the conditional distribution as the product of
its coordinate-wise marginals:
\begin{align*}
\pt^{(1)}_{\psi_i}(\matr{Z}_i, \matr{W}_i) & \triangleq \pt_{\psi_i}(\matr{Z}_i) 
\pt_{\psi_i}(\matr{W}_i) \\ & = \otimes_{j=1}^p \pt_{\psi_i}(\matr{Z}_{ij})
\pt_{\psi_i}(\matr{W}_{ij}).
\end{align*}
On top of that, we assume Gaussian and Bernoulli distribution for $\pt_{\psi_i}(\matr{Z}_{ij})$ and
$\pt_{\psi_i}(\matr{W}_{ij})$ respectively, giving rise to the following variational approximation
\begin{equation} \label{eq:var-prox-1}
\ptun_{\psi_i} \left(\boldsymbol Z_i, \boldsymbol W_i \right) = \otimes_{j=1}^p \mathcal N\left( M_{ij}, S_{ij}^2\right) \mathcal B\left(P_{ij} \right)
\end{equation}
with $0 \leq P_{ij} \leq 1$ and $\psi_i = \left(M_{ij}, S_{ij},
P_{ij}\right)_{1 \leq j \leq p} $. We denote $\boldsymbol M, \boldsymbol S$ and
$\matr P$ the $n\times p$ matrices with respective entries $M_{ij}, S_{ij}$
and $P_{ij}$ ($1 \leq i \leq n, 1 \leq j \leq p$). This approximation therefore requires the estimation of $3 n p$ additional variational parameters on top of $\theta$.

\paragraph*{Enhanced variational approximation}

As $W_{ij}$ can take only two values, the dependence between $\matr Z_{ij}$ and
$\matr W_{ij}$ can easily be made explicit by noting that
\begin{multline} \label{eq:conditioning} 
    Z_{ij} | W_{ij}, Y_{ij} = \\ \left(Z_{ij}|Y_{ij}, W_{ij} 
    = 1 \right)^{ W_{ij}}\left(Z_{ij}|Y_{ij}, W_{ij} = 0 \right)^{1-
W_{ij}}.\end{multline}
The conditional distribution of $Z_{ij}|Y_{ij}, W_{ij} = 1$ simplifies to
$Z_{ij}| W_{ij} = 1$ and is thus known as $Z_{ij}$ and $W_{ij}$ are independent:
it follows a Gaussian distribution with mean $ \boldsymbol x_i \tr B_j$ and variance
$\Sigma_{jj}$. By contrast, $Z_{ij}| Y_{ij}, W_{ij}=0$ is untractable and approximated
by a Gaussian distribution, giving rise to an alternative and slightly more involved variational approximation:
\begin{multline}  \label{eq:var-prox-2}
    \ptdeux_{\psi_i}(\boldsymbol Z_i, \boldsymbol W_i) = \\
    \otimes_{j = 1}^p \mathcal N(\boldsymbol x_i \tr \matr B_j, \Sigma_{jj})^{W_{ij}} \mathcal N(M_{ij},
    S_{ij}^2)^{1-W_{ij}} W_{ij}, \\ 
    W_{ij} \sim^\text{indep} \mathcal B\left(P_{ij}\right).
\end{multline}
The resulting ELBOs are tractable and detailed below.

\begin{remark} \label{rem:easy}It can be shown that $W_{ij}| Y_{ij},Z_{ij} \sim\mathcal B \left( \sigma \left(
\log\left(\frac{\pi_{ij}}{1- \pi_{ij}}\right) + Z_{ij}
    \right)\right) \boldsymbol 1_{Y_{ij} = 0}$ where $\sigma(\cdot) = \logit^{-1}$,
    so that one could condition the other way around. While this
    conditional law is intuitive and appears simpler than \Cref{eq:conditioning} at first glance, as it only
    involves a Bernoulli variable, it turns out to be untractable. Indeed, the resulting ELBO
    involves the entropy term $\mathbb{\tilde E} \left[\log
    \tilde p_{\psi}(\matr W | \matr Z) \right]$, the computation of which requires computing expectations of the form
    $\mathbb{E}\left[ \log\left( \sigma\left(U\right)
\right)\sigma\left(U\right) \right]$ for arbitrary univariate Gaussians $U \sim \mathcal N \left(  \mu,\sigma^2 \right)$,
which are untractable when $U$ is non-degenerated.
\end{remark}

\subsection{Expected lower bounds} We set $\psi = \left(\psi_i \right)_{1 \leq i
\leq n}$ the variational parameters of the variational distribution $\ptun_{\psi} = \prod_{i=1}^n \ptun_{\psi_i}$
 (resp. $\ptdeux_{\psi} = \prod_{i=1}^n \ptdeux_{\psi_i}$) defined in
\Cref{eq:var-prox-1} (resp. \Cref{eq:var-prox-2}). We denote by $\mEtun$ (resp. $\mEtdeux$) its expectation and by
$\Jun(\psi, \theta)$ (resp. $\Jdeux(\psi, \theta)$) the corresponding ELBO, the expression of which is detailed in the next proposition.

\begin{restatable}{proposition}{elbos} \label{prop:elbo}
    The ELBO defined in \Cref{ELBO_theorique} with variational approximation
    $\ptun_{\psi}$ can
    be written in matrix form as
\begin{align}
    \Jun & (\psi, \ta) = \mEtun \left[\log p_{\ta}(\matr Y | \matr Z, \matr W) \right] + \mEtun \left[\log p_{\ta}(\matr W) \right] \nonumber \\
    & + H(\matr P)  + \frac 1 2 \trace \left(\matr 1_{n,p}\tr \log (\matr{S^2}) \right) +  \frac{n}{2} \logdet{\Omega} \label{eq:ELBOMF} \\
    & - \frac 1 2  \trace \left(  \matr \Omega \left( \Diag(\bar{\matr S}^{2}) +  g\left(\matr M - \matr X \matr B \right)\right) \right) + \frac{np}{2} \nonumber
\end{align}
and with variational approximation $\ptdeux_{\psi}$ we get
\begin{align}
    \Jdeux &  (\psi, \ta) = \mEtdeux \left[\log p_{\ta}(\matr Y | \matr Z, \matr W) \right] + \mEtdeux \left[\log p_{\ta}(\matr W) \right] \nonumber\\ 
    & +  H(\matr P)  + \frac 1 2 \trace \left( \matr Q\tr\log (\matr{S^2}) \right)+ \frac{n}{2} \logdet{\Omega}\nonumber\\
    & - \frac 1 2  \trace \left( \matr \Omega \left( \Diag\left(\matr 1_n \tr \left(\matr Q \odot \matr{S}^2\right)\right)\right)\right) \nonumber \\
     & - \frac 1 2  \trace \left(g\left(\matr Q \odot \left(\matr M - \matr X \matr B \right) \right)\right)  \label{eq:ELBO_enhanced} \\
    &  - \frac{1}{2} \trace \left( \diag(\matr{\Omega})1_n \tr\left( (1_n \diag(\matr{\Sigma})\tr) \odot \matr{P}\right)  \right)\nonumber \\
    &  - \frac{1}{2} \trace \left( \diag(\matr{\Omega})1_n \tr\left(\matr{P} \odot \matr Q \odot (\matr{M} - \matr X \matr B)^2 \right)  \right)\nonumber \\
    & - \frac{1}{2} 1_n\tr \matr{P} \log(\diag(\matr{\Sigma})) + \frac{np}{2}\nonumber ,
\end{align}
where $\odot$ denotes the Hadamard product, $\diag$ returns a vector
made of the diagonal of the input squared matrix, $\matr 1_n$ (resp. $\matr 1_{n,p}$) is a column-vector (resp. matrix) of size $n$ (resp. $n\times p$) filled with ones, $\Diag$ takes a vector
$x$ and returns a diagonal matrix with diagonal $x$,
logarithm and squared functions are applied component-wise, $\matr Q = \matr
1_{n,p} - \matr P$  and  $g(\matr D) =  \matr D \tr  \matr D$ for $\matr D \in
\mathbb R^{n \times p}$. We denoted $\bar{\matr S}^{2} = \matr 1_n\tr\matr
S^2$ and
$$\delta_{0, \infty}(x) = \begin{cases} 0 & \text{ if } x = 0 \\ -\infty &
\text{ else} \end{cases}$$ with the convention that $0 \times \delta_{0,
\infty}(x) = 0$ for all $x$. Note that both ELBOs share the following terms ($\mEtun$ and $\mEtdeux$ coincides for these terms so that we drop the index):
\begin{align*}
     \tilde{\mE} & \left[\log p_{\ta}(\matr Y | \matr Z, \matr W) \right] = 
     \trace \left(\matr P \tr \matr \delta_{0, \infty}\left( \matr Y \right)\right)\\
     & \trace \left( \matr Q \tr \left( \matr Y
    \odot \left(\matr O + \matr M\right) - \matr A - \log(\matr Y!) \right)
  \right), \\
     & \tilde{\mathbb E} \left[\log p_{\ta}(\matr W) \right] = \trace \left( \matr
     P \tr \matr  \mu_0 - \matr 1_{n,p}\tr \log\left( \matr{1}_{n,p} + e^{\matr \mu_0}\right) \right)
    , \\
        &  H(\matr P) = -\trace \left(
     \matr{P} \tr \log(\matr{P}) + \matr Q\tr \log\left(\matr Q \right) \right) ,
\end{align*}
where factorial and exponential are applied component-wise and the matrix $\matr
A$ denotes $\exp(\matr O + \matr M + \matr S^2 / 2)$ where $\exp$ is applied
component-wise and $\matr \mu_0 = \matr 1_{n,p}\times\logit(\pi)$ in the ND case,
$\matr \mu_0 =\boldsymbol X^0 \boldsymbol
B^0$ in the CD case and $\matr \mu_0 = \bar{\boldsymbol B}^0 \bar{\boldsymbol
X}^0$ in the RD case. We used the convention $0\times \log(0)=0$ for all $x$.
\end{restatable}

\begin{remark}
The main and only goal of $\delta_{0, \infty}$ is to ensure that $P_{ij} = 0$
whenever $Y_{ij} \neq 0$, \emph{i.e.} that $Y_{ij}$ doesn't originate
from the null component when it's positive.
\end{remark}

\paragraph*{Model selection criterion}

When the modelling choice of zero-inflation is unclear, we consider
two classical criteria: Bayesian Information Criterion or BIC \citep{BIC} and Akaike Information Criterion or AIC \citep{AIC} to choose between
Models \ref{eq:zi-models-nd}, \ref{eq:zi-models-cd} and \ref{eq:zi-models-rd}, where the log-likelihood is replaced by its lower bound $J$. It is worth noting that each of these criteria theoretically requires the
true log-likelihood rather than a surrogate such as the ELBO. \review{Although the ELBO and the log-likelihood may differ substantially in some models, up to the point that variational estimates are inconsistent \citep{sandwich}, the ELBO has been proven to be asymptotically equivalent to the log-Likelihood in others \citep{Brault2020}}. \review{Based on simulations studies, there is strong numerical evidence suggesting that the variational estimates are asymptotically unbiased for the ZIPLN \citep{SandwichBatardiere} and PLN models \citet{chiquet2021poisson} and that the ELBO closely approximates the true likelihood at the optimum \citet{stoehr2024composite}}. Hence, our heuristical versions of BIC and AIC criteria are given by
\begin{align*}
    \text{BIC}^{(i)} & = J^{(i)} - \frac 1 2 K\log(n), \\
    \text{AIC}^{(i)} & = J^{(i)} - K,
\end{align*}
where $K = p(p+1)/2 +pd+ c$ is the number of parameters and $c$ depends
on the modelling choice for the zero-inflation component: $1$ for Model
\ref{eq:zi-models-nd}, $pd_0$ for \ref{eq:zi-models-cd} and $nd_0$ for
 \ref{eq:zi-models-rd}.

We also consider the Integrated Complete Likelihood or ICL \citep{ICL}, a criterion
introduced in the context of mixture models to select the number of components but that still applies when the latent
variable is continuous, as in (Zero Inflated) Poisson LogNormal model. We
recall that ICL uses the conditional entropy of the latent variables given the
observations as an additional penalty with respect to BIC. The difference
between BIC and ICL measures the uncertainty of the representation of the
observations in the latent space. Note that here, ICL is not used as a
criterion for choosing the number of groups but for assessing the uncertainty
of the latent variables $Z$ and $W$.  We believe that ICL, by accounting for
these uncertainties by removing entropies (especially in the context of
variational inference), can lead to interesting model selection output that
complement the BIC or AIC choices. Because the true conditional distribution
$p_{\ta}(\matr Z, \matr W |\matr Y )$ is intractable, we replace it with its
variational approximation $\tilde
p_{\psi}(\matr W, \matr Z)$ to evaluate this entropy. Recalling that $\psi =
\left(\matr M, \matr S, \matr P\right)$, the entropy for variational approximations $\ptun_{\psi}$ and $\ptdeux_{\psi}$ are respectively given by
\begin{align*}
    & H^{(1)}(\psi) = \frac 1 2 \matr 1_n\tr \log\left(\matr S^2\right)\matr 1_p + \frac {np}{2} \log(2\pi e) + H(\matr P) \\
    & H^{(2)}(\psi) = \frac 1 2 \matr 1_n \trace\left(\left( \matr 1_{n,p} - \matr P \right) \tr \log(\matr S^2) \right) \matr 1_p \\
    & \qquad - \frac 1 2 \matr 1_n \tr \matr P \log\left( \diag\left( \Sigma \right) \right) + H(\matr P)
\end{align*}
where $H(\matr P)$ is defined in \Cref{prop:elbo}.  The ICL for variational approximation $\pt_{\psi}^{(i)} ~ (i = \{1, 2\})$ is then
\begin{equation*}
    \text{ICL}^{(i)}  = \text{BIC}^{(i)} -  H^{(i)}(\psi).
\end{equation*}

\begin{remark}
    We note that \ref{eq:zi-models-rd} is not a parametric model, for which AIC, BIC and ICL have a theoretical grounding, but a semi-parametric one. We nevertheless use those criteria to compare the three models to each other.
\end{remark}

The following section discusses the optimization of both ELBOs to estimate
the model parameters $\theta$.

\section{Optimization}
\label{sec:inference}

Estimating $\theta$ is equivalent to solving the optimization problem
\begin{align}
    \label{eq:criterion}
   & \argmax_{\psi, \theta} J( \psi, \ta).
\end{align}
where $J$ can be either $\Jun$ (standard approximation) or $\Jdeux$ (enhanced approximation).

\subsection{Optimization of $\Jun$}
Past experience for standard PLN models
\citep{Chiquet2017,chiquet2019variational, chiquet2021poisson} (and analytical
properties of $\Jun$ derived in this section) suggests solving the above
problem using alternated gradient descent.

Consider the ELBO $\Jun$ defined in \Cref{prop:elbo}. The
following technical propositions will serve to update some parameters
in an alternate optimization scheme.

\begin{restatable}{proposition}{updateBOmegaPB0}[Updates of $\matr B,
    \matr{\Omega}, \matr P$ and $\matr B^0$] \label{prop:updates}
For fixed $\psi$, the values of $\matr{\Omega}, \matr{B}$ maximizing $\Jun$ are
\begin{align*}
    \widehat{\matr{B}} & = [\matr{X}^\top  \matr{X}]^{-1} \matr{X}^\top\matr{M},\\
    \widehat{\matr{\Omega}} & = n \left[ g(\matr{M} - \matr{XB})+ \bar{\matr{S}}^2 \right]^{-1}.\\
\end{align*}
where $g(\matr{D}) = \matr{D}^\top \matr{D}$ as in \Cref{prop:elbo}.
Furthermore, if $\matr X^0 = \matr 1_n$ , $\Jun$ is maximized at
$$\widehat{\matr B}^{0} = \frac 1 n \matr 1_n \tr \matr P.$$
\newline \noindent
When $\ta$ is fixed, $\Jun$ is concave with respect to $\matr P$ and maximized at
$$\widehat {\matr P} = \logit^{-1} \left( \matr{A} + \matr{X}^0\matr{B}^0 \right)
            \times \delta_0(\matr{Y}).$$
\end{restatable}

\begin{proof} 
    Proofs for $\widehat{\matr{\Omega}}$ and $\widehat{\matr{B}}$ and $\widehat{\matr{B}}^0$
    are straightforward using null-gradient conditions given in the appendix (\Cref{prop:J-gradient-mod-param-1}) and left to the reader. We only prove the concavity with respect to $\matr P$. As $\Jun$ is separable in each $P_{ij}$, we only need to prove
    concavity with respect to each $P_{ij}$. If $Y_{ij} > 0$, $P_{ij}\delta_{0,\infty}(Y_{ij}) = -\infty$ as soon as $P_{ij} > 0$ and $\Jun$ is therefore concave in $P_{ij}$. If $Y_{ij} = 0$, $\Jun$ depends on $P_{ij}$ only through $P_{ij} A_{ij} + P_{ij} 
    [\matr{x}_i^{0^\top} \matr{B}_j^0 - \logit(P_{ij})] - \log(1 -P_{ij})$ which is concave in $P{ij}$.
\end{proof}

\SetKwBlock{mstep}{M-step}{}
\SetKwBlock{vestep}{VE-step}{}

\begin{algorithm*}
    \caption{VEM
    \label{alg:VEM}
    }
    \Input{$\ta^{(0)}, \psi^{(0)}$ initial point, $T\geq1$ number of iterations.}
    \For{$s = 0, \dots T-1$}
    {
        \mstep{
        \begin{flalign*} & \matr{\Omega}^{(s+1)} = n \left[ g\left(\matr{M}^{(s)} - \matr{XB}^{(s)}\right) + \bar{\matr{S}^2}^{(s)} \right]^{-1} &&\\
            & \matr{B}^{(s+1)} = [\matr{X}^\top \matr{X}]^{-1} \matr{X}^\top \matr{M}^{(s)} && \\
            & \matr{B}^{0, (s+1)} = \argmax_{\matr{B}^0} \trace\left[\
      \left(\matr{P}^{(s)}\right)^\top \matr{X}^{0}\matr{B}^0 \right] - \trace\left[
        \matr{1}_{n,p}^\top \log\left(1 + e^{\matr{X}^0\matr{B}^0}\right)\right]&&
    \end{flalign*}}
    \vestep{
    \begin{flalign*}
         & \matr{P}^{(s+1)} = \logit^{-1}\left( \matr{A}^{(s)} + \matr X^0\matr{B}^{0, (s+1)}\right) \times \delta_{0}(\matr{Y}),\quad  \matr Q^{(s+1)} = \matr 1_{n,p} - \matr P^{(s+1)} && \\
         & \matr{M}^{(s+1)} = \argmax_{\matr{M}} \left(
             \trace\left(\matr Q^{(s+1)\top} \left(\matr{Y} \odot
\matr{M} - \matr{A}\right)\right) - \frac12 \trace\left( \matr{\Omega}^{(s+1)} g\left(\matr{M} -
\matr{XB}^{(s+1)}\right) \right) \right) && \\
                                         & \matr{S}^{(s+1)} = \argmax_{\matr{S}} \left( -\trace\left(Q^{(s+1) \top} \matr{A}\right) - \frac12
\trace\left(\matr{1}_{n,p}^\top \log\left(\matr{S}^2\right)\right) - \frac12
    \trace\left(\matr{\Omega}^{(s+1)} \bar{\matr{S}}^2\right) \right)&&
\end{flalign*}}
}\
    \Output{$\ta^{(T)}, \psi^{(T)}$}
\end{algorithm*}

An alternated gradient descent optimizing $\Jun$ is proposed in \Cref{alg:VEM} and convergence to a stationary point is a direct consequence of the following lemma.
\begin{restatable}[Convergence properties]{lemma}{concavity}
    $\Jun$ is (separately) concave in $\theta$ and $\psi$.
\end{restatable}
\begin{proof}

 For $\psi = (\psi_1, \matr P)$ with $\psi_1 = (\matr M, \matr S)$, note first that $\Jun$ is separable in $\psi_1$ and
 $\matr P$ so we can prove it independently for each parameter. For $\psi_1$, it
 follows from the same result in the standard PLN model (lemma 1 of \citet{Chiquet2017}). For $\matr P$, it follows from
 \Cref{prop:updates}.
 For $\theta$, note
 first that $\Jun$ is separable in $(\matr{B}, \matr{\Omega})$ and $\matr{B}^0$ so
 we can prove it independently for each parameter. For the former, it follows
 from the same result in the standard PLN model (for $(\matr{B},
 \matr{\Omega})$).
For the latter, it follows from the concavity (for all $\matr a \in \mathbb R^d$ fixed)
  of the following function $f: \mathbb R^d \mapsto \mathbb R$:
 \begin{equation*}
  f(\matr{\beta})  = \matr{a}^\top \matr{\beta} -
  \matr{1}_n^\top\log \left( 1+e^{\matr{X}^0\matr{\beta}}\right).
\end{equation*}
\end{proof}

\subsection{Optimization of $\Jdeux$}
While optimization of $\Jun$ is easily manageable using closed forms  and
benefits from a bi-concavity property, optimization of $\Jdeux$
is more challenging. Indeed, the concavity in
$\matr \Omega$ is lost and no closed form can be used for any parameter update.

We do not maximize the ELBO with respect to each parameter in an alternate coordinate-wise fashion but instead compute the gradient with
respect to $(\psi,\ta)$ as if it were a single parameter and perform a gradient
update. Formally, given $\psi^{(0)}, \ta^{(0)}$ and a learning rate $\eta>0$, we perform the update step
\begin{equation}\label{eq:upate_J2}
(\psi^{(s+1)}, \ta^{(s+1)}) = (\psi^{(s)}, \ta^{(s)}) + \eta \nabla_{\psi, \ta}\Jdeux( \psi^{(s)}, \ta^{(s)})
\end{equation}
until a convergence criteria or a maximum number of iterations is reached.
    \subsection{Optimization using analytic law of $W_{ij}| Y_{ij}$}
    The exact conditional law $W_{ij} |
    Y_{ij}$  can be derived and is detailed in the next proposition.
\begin{restatable}{proposition}{conditionals}
    \label{prop:law_W}
    Let $1 \leq j \leq p$. The conditional law of $W_{ij}| Y_{ij}$ is given by
\begin{equation*}
    \mathcal B\left(\frac{\pi_{ij}}{ \varphi\left(o_{ij} + \matr x_i\tr \boldsymbol B_j, \Sigma_{jj}\right)
\left(1 - \pi_{ij}\right) + \pi_{ij}}\right) \boldsymbol 1_{Y_{ij} = 0}
\end{equation*}
with $\varphi(\mu,\sigma^2) = \mathbb E \left[ \exp(-X)\right], ~ X \sim
\mathcal L \mathcal N \left( \mu, \sigma^2\right)$.
\end{restatable}
In \Cref{sec:var-prox} we made the variational approximation $\pt(W_{ij})
    \sim \mathcal B (P_{ij})$, considered $P_{ij}$ as free and optimized the
    ELBO with respect to $P_{ij}$.  The above proposition suggests that $P_{ij}$
    can instead be derived directly from $\ta$ and not be considered as a free
    variational parameter.  We consider $\tJun$ (resp. $\tJdeux$) the ELBO $\Jun$ (resp. $\Jdeux$) with
    $P_{ij} = \Psi(\ta)_{ij}$ with
    \begin{multline*}
        \Psi(\ta) \triangleq \\ \frac{\matr \pi}{
    \varphi\left(\matr O + \matr X\tr \matr B, \matr 1_n \diag(\matr \Sigma)\tr \right)\odot \left(1 -
\matr \pi\right) + \matr \pi}\odot \boldsymbol 1_{\boldsymbol Y = \matr 0},
    \end{multline*}
where $\varphi$ and the division are applied component-wise and $\boldsymbol 1_{\boldsymbol Y
    = \matr 0}$ is a $n\times p$ matrix such that $\left(\boldsymbol
    1_{\boldsymbol Y } = 0\right)_{ij} = 0$ if and only if $Y_{ij} = 0$.
Formally, we have
\begin{multline*}
    \tJun(\matr M, \matr S,\matr \Omega, \matr B, \matr B^0) = \\ 
    \Jun\left(\matr M, \matr S, \Psi(\matr \Omega, \matr B, \matr B^0), \matr \Omega, \matr B, \matr B^0\right),
\end{multline*}
and the same formula applies to $\tJdeux$. Note that both $\tJun$ and $\tJdeux$
have $np$ fewer variational parameters compared to $\Jun$
and $\Jdeux$ ($2np$ compared to $3np$) since $\matr P$ is now completely determined by $\ta$. The function $\varphi$ is intractable but a sharp (derivable)
approximation $\tilde \varphi$ is available and detailed in the next section. For the Standard approximation, a
major drawback of this approach compared to optimizing $\Jun$ is the lack of any
closed form update as stationary points of $\tilde \varphi$ are
intractable. For the optimization, we consider the gradient scheme defined in
\Cref{eq:upate_J2} where $\psi$ is replaced with $\psi_1 = \left(\matr M,
\matr S \right)$.
\subsection{Implementation details}
\paragraph*{Gradient with respect to $\matr \Omega$ }
When no closed form is available, optimization with
respect to $\matr \Omega$ must be adapted to ensure that $\matr \Omega$ remains symmetric and
positive definite. Instead of maximizing directly over $\matr \Omega$, we introduce a $p \times p$ unconstrained
matrix $\matr C$ and use the following parametrization for $\matr \Omega$
\begin{equation} \label{eq:Omega_trick}
    \matr \Omega = (\matr C \matr C\tr) ^{-1}
\end{equation} and compute
the gradient with respect to $\matr C$.
\paragraph*{Approximation of $\varphi$}
The function $\varphi$ defined in \Cref{prop:law_W} is intractable but an approximation \citep{lognormcharact} can be computed:
\begin{equation*}
    \varphi(\mu, \sigma^2)\approx \tilde \varphi(\mu, \sigma^2)=  \frac{\exp \left(-\frac{W^2\left(\sigma^2
    e^\mu\right)+2 W\left(\sigma^2 e^\mu\right)}{2
\sigma^2}\right)}{\sqrt{1+W\left(\sigma^2 e^\mu\right)}},
\end{equation*}
where $W(\cdot)$ is the Lambert function (i.e. $z = x \exp(x) \Leftrightarrow x = W(z),
~~ x,z \in \mathbb R$). An analysis of its sharpness is performed in \cite{sharpvarphi}. Derivability of $\tilde \varphi$ is ensured as $W(\cdot)$ is derivable.
\paragraph*{Stochastic Gradient Ascent}

When $n$ is large, computing the whole gradient $\nabla_{\psi} J$ is
time-consuming. As the ELBOs defined in \Cref{prop:elbo} are additive in the
variational parameters $\psi$, Stochastic Gradient Ascent \citep{SGD} can be applied to scale the algorithm to large datasets.

\section{Simulation Study}
\label{sec:sim-study}
In this section, we evaluate the statistical and computational performances of Models \ref{eq:zi-models-nd}, \ref{eq:zi-models-cd} and \ref{eq:zi-models-rd} on simulated data.
\subsection{Experimental details}\label{sec:expdetails}
We set $n = 1000, p = 250, d =3 $ (and $d_0=4$ for \ref{eq:zi-models-cd},\ref{eq:zi-models-rd}) to mimick typical sizes observed in microbiome studies and in particular in our case study. Given $\ta^{\star} = \left( \matr
\Sigma^{\star}, \matr {B}^{\star}, \matr{\pi}^\star \right)$  for Model \ref{eq:zi-models-nd}, $\ta^{\star} = \left( \matr
\Sigma^{\star}, \matr {B}^{\star}, \matr B^{0^\star} \right)$  for Model \ref{eq:zi-models-cd} and $\ta^{\star} = \left( \matr
\Sigma^{\star}, \matr {B}^{\star}, \matr {\bar{B}}^{0^\star} \right)$  for Model \ref{eq:zi-models-rd}, we
simulate $n$ independent observations $\matr Y_i$ for each model, and consider the following estimation strategies:
    \begin{itemize}
        \item Standard $(\Jun)$,
        \item Enhanced $(\Jdeux)$,
        \item Standard Analytic $(\tJun)$,
        \item Enhanced Analytic $(\tJdeux)$,
        \item PLN,
        \item Oracle PLN,
    \end{itemize}
    where PLN competitor corresponds to a PLN model fitted directly on the zero-inflated data
    matrix $\matr Y$ whereas Oracle PLN is fitted on the non-inflated data $\matr T$, where $T_{ij}|Z_{ij} \sim \mathcal P \left(  \exp(O_{ij} + Z_{ij}) \right)$. While the former evaluates model performance without
    considering zero-inflation during modeling, the latter serves as a
    reference for the Poisson component, as it is unaffected by the signal
    degradation caused by zero-inflation.
To assess inference quality, we report the following metrics:
    \begin{itemize}
        \item Root Mean Squared Error (RMSE) between true parameters $\ta^{\star}$ and estimates $\widehat \ta$, as well as between true zero-inflation probabilities $\matr{\pi}^{\star}$ and estimated probabilities $\hat{\matr \pi}$, computed as $\frac{1}{N} \sum_{i=1}^N \|\widehat{\ta}_i - \ta^{\star}_i\|_2$ and $\frac{1}{N} \sum_{i=1}^N \|\hat{\matr \pi}_i - \matr{\pi}^{\star}_i\|_2$ where $N$ is the number of simulations for each simulation setting;
        \item Reconstruction error, computed as the RMSE between the original data matrix $\matr Y$ and the reconstructed data matrix $\widehat{\matr Y}$ or $\frac{1}{N} \sum_{i=1}^N \|\hat{\matr Y}_i - \matr{Y}_i\|_2$ where $N$ is the number of simulations for each simulation setting;
        \item ELBO;
        \item Computation time.
    \end{itemize}
    The results pertaining to the reconstruction error and computation are deferred to \Cref{sec:appendixB} for detailed analysis.

    We investigate how the models respond to fluctuations in the zero-inflation
    probability $\matr \pi^{\star}$ (\Cref{ssec:pi_fluct}) and fluctuations in the mean
    $\matr X \matr B^{\star}$ of the Gaussian component $\matr Z$
    (\Cref{ssec:xb_fluct}). In the former, we explore whether an increase in
    zero-inflation probability results in better estimation of the
    zero-inflation parameter $\matr{\pi}^\star$ (or $\matr B^{0^\star}$, $\matr
    B^{0^\star}$) at the expense of degraded estimation for the PLN parameters
    $(\matr{\Sigma}^\star, \matr B^{\star})$. In the latter, we assess the
    model accuracy in challenging scenarios where $\matr X \matr B^{\star}$
    is small, leading to numerous zeros in both the ZI and PLN components.
    Additionally, in \Cref{ssec:n_fluct}, we examine the performance
    enhancement as the sample size $n$ increases.

The covariance matrix $\matr \Sigma^{\star}$ is defined as a random Toeplitz matrix:
\begin{equation*}
    \matr \Sigma_{kj}^{\star} = \alpha^{|j-k|}, \, \alpha \sim \mathcal U([0.7,0.9]), \, (j,k) \in \{1,\dots,p\}^2,
\end{equation*}
where $\mathcal U$ stands for the uniform distribution.
The random parameter $\alpha$ controls the amount of correlation between variables.
The \( n \times d \) parameters \(\matr{X}\), \(\matr{X}^0\), and
\(\bar{\matr{B}}^{0^{\star}}\) consist of independent rows (denoted as \( x_i
\)), each following a multinomial distribution where \(\mathbb{P}(x_{ir} = 1) =
1/d\) for \( r \in \{1, \dots, m\} \). Notably, none of these matrices includes
an intercept term.
All entries of $\matr B^{\star}$ are independently sampled from a Gaussian 
distribution with a mean $\gamma \in \mathbb R$ and a variance of $1$, ensuring 
that $\matr X \matr B^{\star}$ exhibits independent Gaussian entries 
centered on $\gamma$ with unit variance.Parameters $\matr B^{0^\star}$ and 
$\bar{\matr X}^0$ follow the same generation process, except that the Gaussian 
mean is determined by $\logit(\rho)$ for some $0 < \rho < 1$. A larger $\rho$ 
corresponds to increased zero-inflation, while a larger $\gamma$ indicates a 
larger Poisson Log-Normal (PLN) component. Specific values for $\rho$ and 
$\gamma$ are provided in each subsection. Notably, offsets ($\matr O$) are not 
considered in these simulations, and are set to a zero matrix of dimensions $n 
\times p$.

\subsection{Simulations when $\matr \pi^{\star}$ fluctuates}\label{ssec:pi_fluct}

The parameter $\gamma$ is assigned a value of $\gamma = 2$ to introduce a
moderately large Poisson Log-Normal (PLN) component, characterized by a low
probability (6.5\%) of generating zeroes. To regulate the degree of
zero-inflation, we systematically increase the probability of zero-inflation
from 0.2 to 0.9, with increments of 0.1. This results in
$\matr \pi^{\star}$ taking values in $\Pi \triangleq \{0.2, 0.3, \dots, 0.8, 0.9\}$.
For Model \ref{eq:zi-models-nd}, this adjustment is straightforward, as
$\matr \pi^{\star}$ is directly set to values within $\Pi$. In the case of Model
\ref{eq:zi-models-cd} (and Model \ref{eq:zi-models-rd}), the control of
zero-inflation is achieved by simulating $\matr B^{0^{\star}}$ (or $\bar{\matr
X}^{0}$), as described previously, with $\rho$ taking values from $\Pi$.
Following this methodology, we generate, for each $\rho$ in $\Pi$, 30 distinct 
parameter sets $\theta$, resulting in a total
of  $8\times 30$  distinct parameter sets $\theta$ for each model. For each
$\theta$, we simulate $\matr Y$ according to \Cref{eq:ZIPLN-model}. We
use this simulation scheme (rather than the alternative of sampling a single
value of $\theta$ for each parameter set and then simulating 30 replicates of $\matr Y$ for
that $\theta$) to average results over values of $\theta$ and avoid sampling an
easy or a difficult $\theta$ by chance, which would blur the trends and make
results harder to interpret. The obtained results are depicted in
\Cref{fig:proba_stat}.

Interpreting the results for $\matr \pi$ poses certain challenges. At low
values of $\matr \pi^{\star}$, the inflated component generates few zeros,
resulting in a clear distinction between zeros from the inflated and Poisson
components, albeit with a lower overall count of zeros from the inflated
component. In contrast, at high values of $\matr \pi^{\star}$, the inflated
component produces more zeros, increasing the overall
prevalence of observed zeros but making it more difficult to distinguish their
origin.

The RMSE of $\matr \Sigma^{\star}$ exhibit the same qualitative results across
models. The Standard Variational Approximation (VA) outperforms other VA
methods but falls short, as expected, of Oracle PLN, which is unaffected by $\matr
\pi^{\star}$ as it works on non-inflated counts.
Standard and all other methods, with the notable exception of Standard
Analytic, exhibit performance degradation as $\matr{\pi}^{\star}$ increases.
All methods appear to perform similarly when $\matr{\pi}^{\star}$ approaches
$1$.

Regarding $\matr \pi^\star$, results from Model
\ref{eq:zi-models-rd} (\Cref{fig:proba_stat}, middle right panel) indicate a progressive decrease in RMSE until $\matr
\pi^{\star}$ reaches the $0.7$ threshold, past which the RMSE starts to increase.
This trend reflects the delicate balance in model performance: when $\matr
\pi^{\star}$ is too small, there are too few observed zeroes to reliably estimate $\matr
\pi^{\star}$, leading to high RMSE. Conversely, high values of $\matr
\pi^{\star}$ result in an overabundance of zeroes and a loss of signal,
adversely impacting model performance. Intermediate values of $\matr
\pi^{\star}$ lead to the best performance, similar to what is observed in
logistic regression. By contrast, results from Model \ref{eq:zi-models-cd} (\Cref{fig:proba_stat},
middle panel) reveal a gradual reduction in RMSE across all methods, except for Standard which
nevertheless demonstrates the best performance except for very high values of $\matr
\pi^{\star}$. In the case of Model \ref{eq:zi-models-nd} (\Cref{fig:proba_stat}, middle left panel),
both Enhanced VA methods exhibit
improved performance as $\matr \pi^{\star}$ increases. However, Enhanced Analytic
reaches a plateau when $\matr \pi^{\star} \geq 0.5$. This plateau corresponds to a
scenario characterized by a high incidence of observed zeroes and a low signal
from the Poisson component, that might justifying the observed stabilization in
performance. Note also that $\matr \pi^{\star}$ is as expected much easier to
estimate in Model \ref{eq:zi-models-nd} than in models \ref{eq:zi-models-cd}
and \ref{eq:zi-models-rd}, as shown by the difference in scales. The RMSE of
$\matr B$ shows a slight increase with $\matr \pi^{\star}$ across all methods,
with the exception of Standard Analytic, just like the RMSE for $\matr \Sigma^{\star}$.
Standard once again matches or surpasses all methods except the oracle PLN.

\begin{figure*}[h]
            \includegraphics[width=\linewidth]{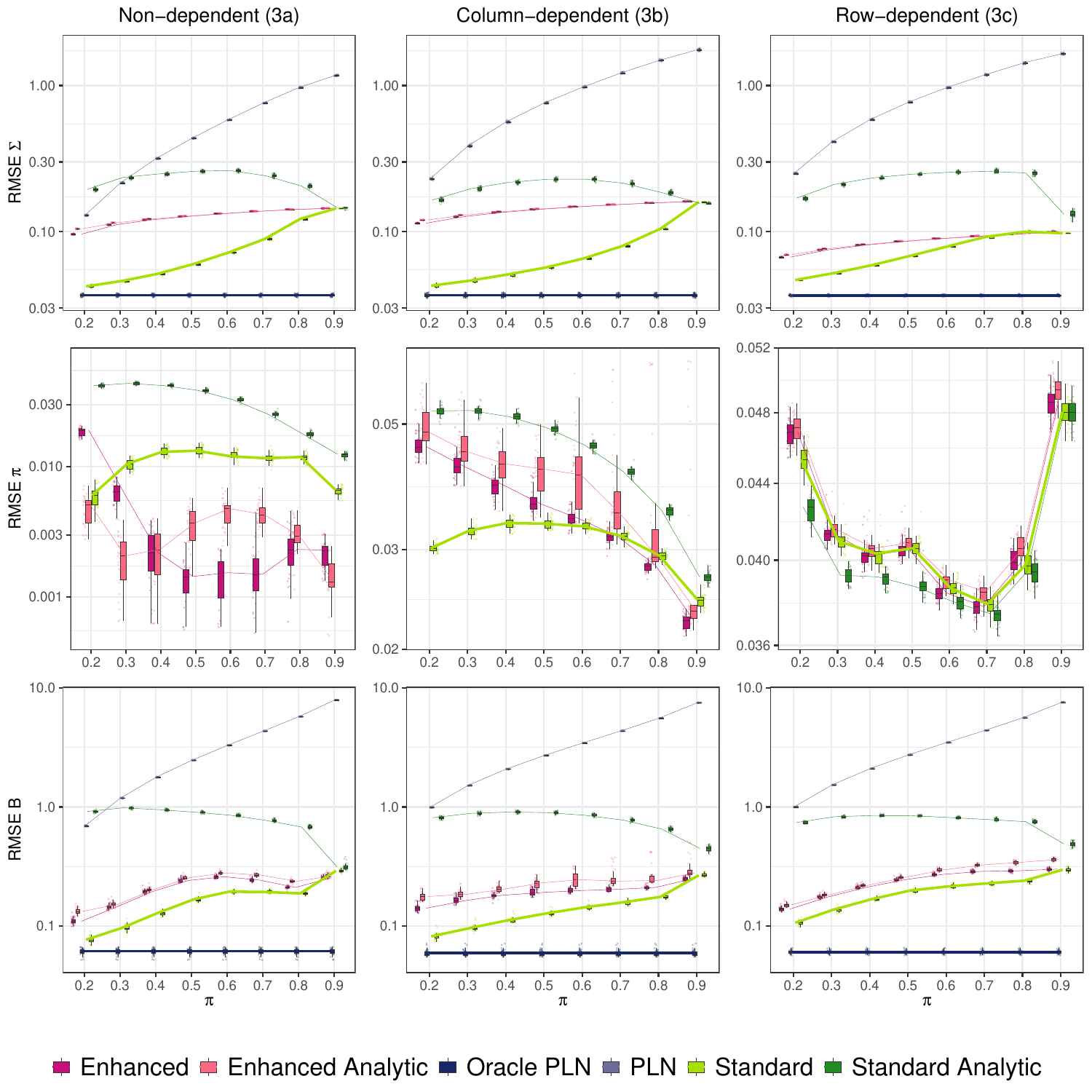}
            \caption{Simulations when the probability of zero-inflation $\matr 
\pi^{\star}$ varies. To improve visibiliy, the scale is \textbf{not} shared on 
the second row. The line corresponding to Oracle PLN (resp. Standard) is 
highlighted as the oracle model (resp. the best performing model).}
            \label{fig:proba_stat}
\end{figure*}

\subsection{Simulations when $\matr X \matr B^{\star}$ fluctuates}\label{ssec:xb_fluct}

In Model \ref{eq:zi-models-nd}, we fix $\matr \pi^{\star}$ at $0.3$, a value chosen
based on its propensity to produce moderately challenging models and yield
contrasting results in previous experiments. For Model \ref{eq:zi-models-cd}
(and Model \ref{eq:zi-models-rd}), we simulate the parameters $\matr
B^{0^{\star}}$ (and $\bar{\matr X}^{0}$) as previously described, setting $\rho
= 0.3$. To systematically enhance the signal of the Poisson component, we increment
$\gamma$ from 0 to 3 in steps of 0.5, thereby covering values in $\Gamma
\triangleq \{0, 0.5, \dots, 2.5, 3\}$. Following this methodology, we generate, 
for each $\gamma$ in $\Gamma$, 30 disinct parameter sets $\theta$, resulting in 
a total
of  $7\times 30$  distinct sets $\theta$ for each model. For each $\theta$, we 
simulate $\matr Y$ according to \Cref{eq:ZIPLN-model}.
The obtained results are depicted in \Cref{fig:poisson_stat}.

Regarding the RMSE of $\matr \Sigma^{\star}$, all methods demonstrate stable performance across models.
Standard and all other methods, with the exception of Standard Analytic and 
Oracle PLN, exhibit increased performance as $\matr X \matr B^{\star}$ 
increases. Likewise, both Standard Analytic and PLN exhibit an increase in RMSE 
when $\matr X \matr B$ ncreases. On the other hand, the RMSE of Standard always 
decreases with $\matr X \matr B$, similarly to Oracle PLN.

This behavior is expected: when the mean \( \matr{X} \matr{B} \) is low, the
signal diminishes significantly, rendering the covariance structure challenging
to observe and estimate due to the uniformly low values.

Regarding $\matr \pi^{\star}$, all methods provide better estimations when 
$\matr X \matr B$ increases. With larger magnitudes of $\matr X \matr B$, the 
Poisson component generates fewer zeros, leading to clearer identification of 
the origin of zeros. All methods yield comparable results; however, Standard 
once again matches or outperforms other methods.

Finally regarding $\matr B^{\star}$, Standard, Enhanced Analytic, and Enhanced exhibit
stable and comparable results across all three models: they provide better
estimates when $\matr X \matr B$ increases, as expected and for the same
reasons as outlined above. Among those, Standard demonstrates slightly better
performance. As expected, Oracle PLN outperforms all methods as it is fitted on
non inflated data and can rely on more information to estimate $\matr B$
whereas PLN performs increasingly worst as the coefficient are biased downward
by the zeroes.


\begin{figure*}[h]
            \includegraphics[width=\linewidth]{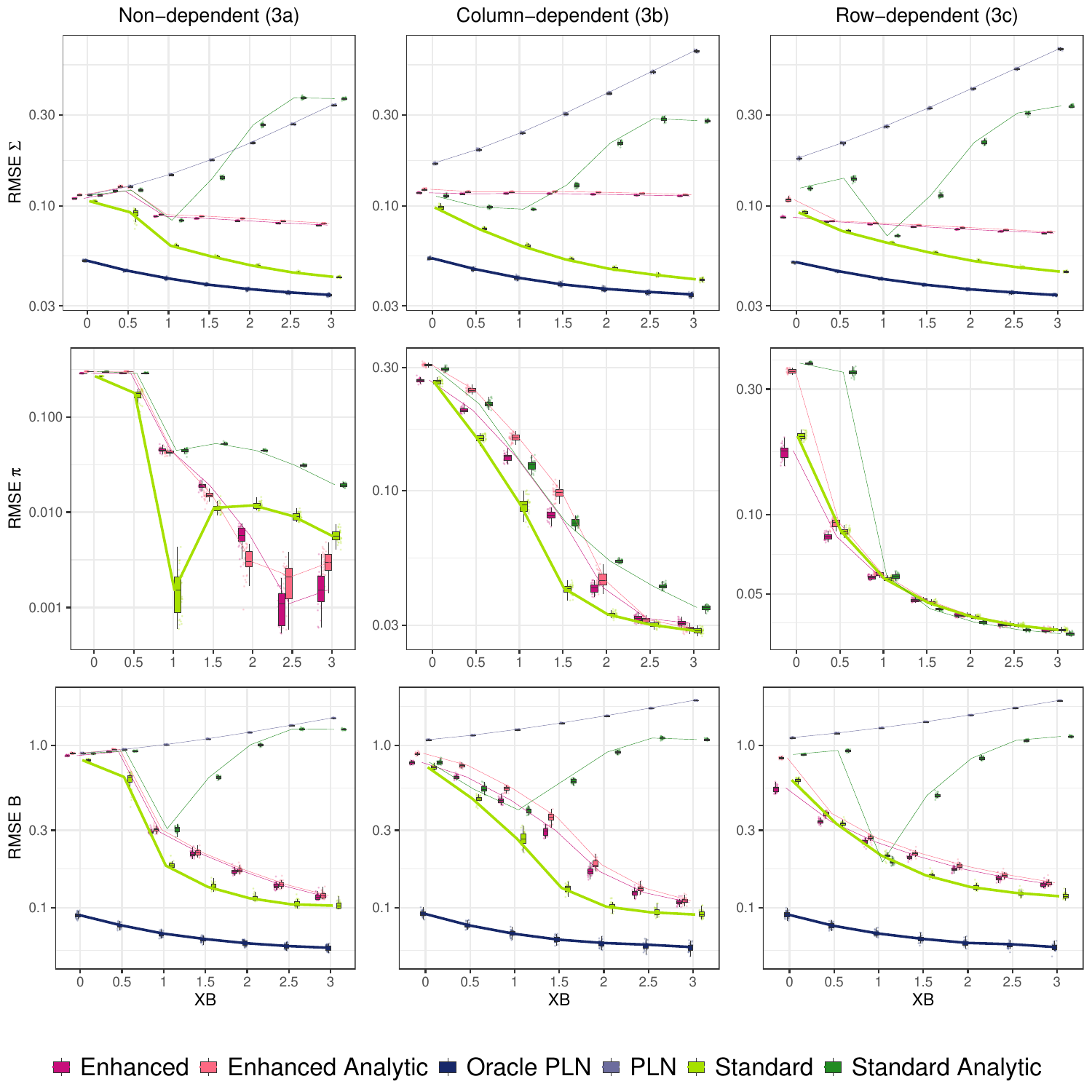}
            \caption{Simulations when the magnitude of the gaussian component 
$\matr X \matr B$ varies. To improve visibiliy, the scale is \textbf{not} shared 
on the second row. The line corresponding to Oracle PLN (resp. Standard) is 
highlighted as the oracle model (resp. the best performing model).}
            \label{fig:poisson_stat}
\end{figure*}
\subsection{Simulations when $n$ fluctuates}\label{ssec:n_fluct}

The parameter $\gamma$ is maintained at $\gamma = 2$, as discussed in Section
\ref{ssec:pi_fluct}, while the zero-inflation parameter $\matr \pi^{\star}$
remains set at $0.3$, consistent with the conditions outlined in Section
\ref{ssec:xb_fluct}. Furthermore, the number of variables is held constant at
$p=250$. To investigate the impact of the sample size, we incrementally vary 
$n$ from $100$ to $600$, with increments of $100$, forming the set
$N={100,200,\dots,600}$. For each value in $N$, we generate $30$ distinct parameter sets $\theta$,
resulting in a total of $30 \times 6$ different parameter combinations for each 
model. Subsequently, we simulate datasets $\matr Y$ for each parameter $\ta$. 
The results are presented in detail in Figure \ref{fig:samples_stat}.

For the Enhanced VA, the addition of more samples does not yield a substantial
improvement in RMSE. Hence, further exploration into this aspect is deemed
unnecessary. Notably, it performs comparably to other methods when $n$ is
low.

The RMSE of $\matr \Sigma^{\star}$ generally decreases, as expected, with $n$,
except for Standard Analytic and at slower rater for Enhanced (analytic
or not).
Similarly, the RMSE of $\matr \pi$ decreases with $n$ across all methods,
with a particularly notable improvement in Model \ref{eq:zi-models-cd}.

In Model \ref{eq:zi-models-nd}, the decrease is less pronounced for all VA
methods, with a plateau after $n \geq 250$.
Analyzing Model \ref{eq:zi-models-rd} with respect to $\matr \pi^{\star}$ presents
challenges due to the increasing number of parameters of $\matr B^0$ with $n$,
influencing $\matr \pi^{\star}$. The RMSE with respect to $\matr B^{\star}$ displays a
consistent decrease for all VA methods, except the Enhanced  and Standard
Analytic ones. Standard Analytic reaches a plateau when the sample size
$n$ reaches 250. Both Standard and Enhanced Analytic exhibit similar 
performances,
outperforming Standard Analytic by a considerable margin.

As a general trend, it appears that the Enhanced Analytic and
Enhanced methods outperform the Standard method on the low samples regime ($n < p=250$).
This can be explained by the fact that these methods are regularized and thus
more robust to overfitting, which is a common issue in low sample regimes.
However, as the sample size increases, the Standard method outperforms the
Enhanced methods. This is likely due to the fact that the Standard method
has more flexibility and can better adapt to the data when there are enough
samples.

\begin{figure*}[btp]
    \includegraphics[width=\linewidth]{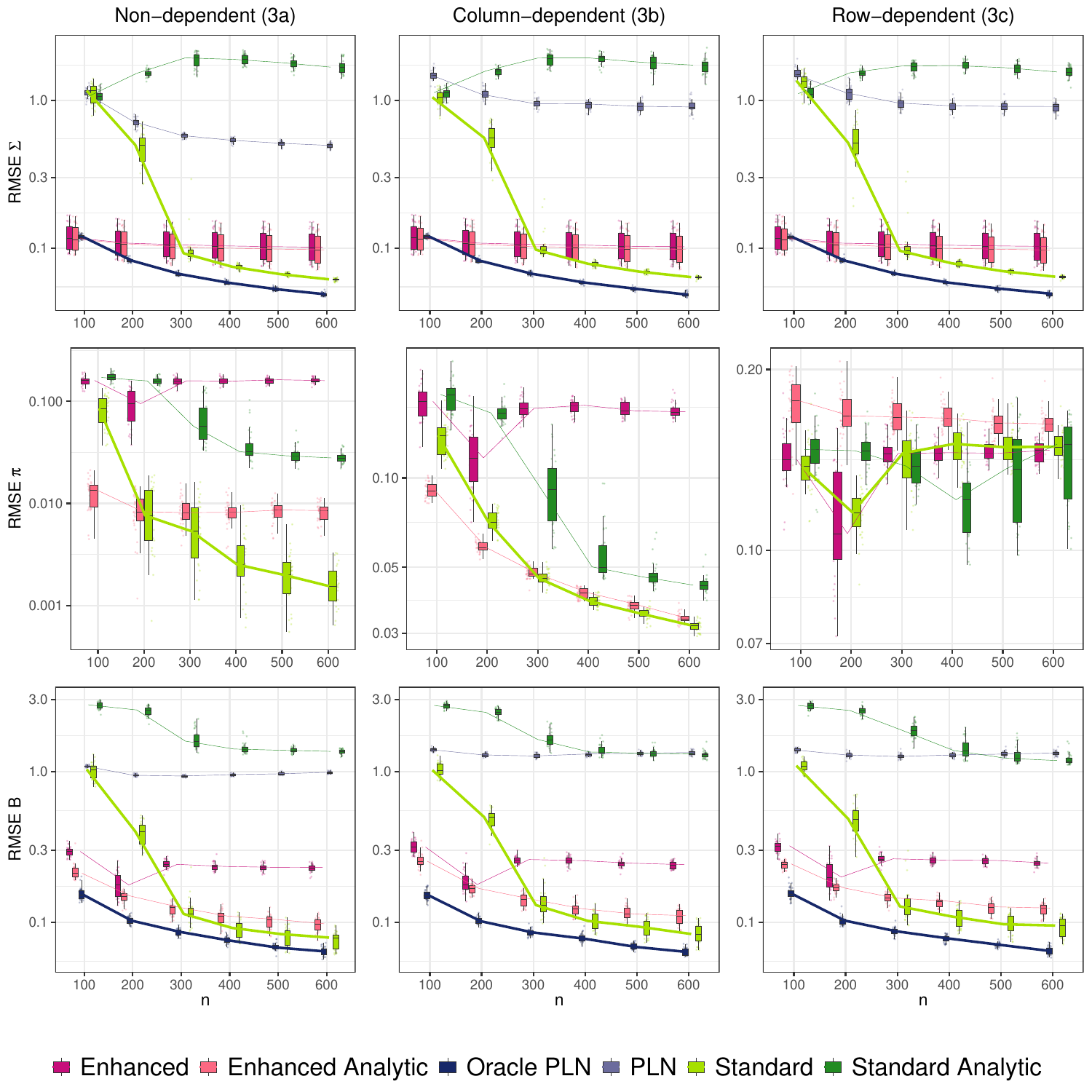}
    \caption{Simulations when the number of samples $n$ grows. To 
improve visibiliy, the scale is \textbf{not} shared on the second row. The line 
corresponding to Oracle PLN (resp. Standard) is highlighted as the oracle model 
(resp. the best performing model).}
    \label{fig:samples_stat}
\end{figure*}

Table \ref{tab:elbos} presents the ELBOs for $n=\{100,300,500\}$, showcasing only the highest-performing methods identified from \Cref{fig:samples_stat,fig:poisson_stat,fig:proba_stat}: Enhanced Analytic
and Standard. The parameter $\gamma$ is fixed at $\gamma = 2$, following the
conditions outlined in Section \ref{ssec:pi_fluct}, while the zero-inflation
parameter $\matr \pi^{\star}$ remains set at $0.3$, consistent with the
specifications in Section \ref{ssec:xb_fluct}. Additionally, the dimensionality
is maintained at $p=250$. For each sample size $n$, 30 distinct parameter sets
$\theta$ are generated, resulting in a total of $30 \times 3$ different
parameter combinations. Subsequently, a dataset $\matr Y$ is simulated for each
parameter $\ta$. The ELBO is computed as the mean across the 30 different runs,
accompanied by a $95\%$ confidence interval.

\begin{table*}[btp]
    \begin{center}
        \includegraphics[width=\linewidth]{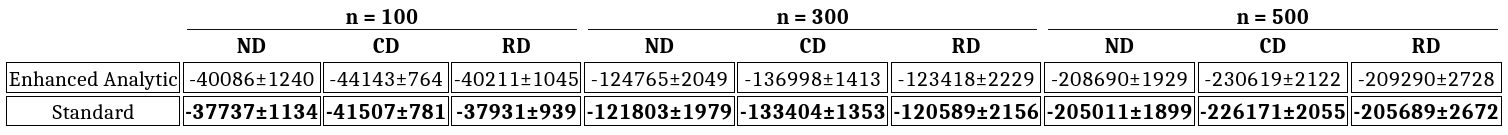}
        \caption{\label{tab:elbos}Comparison of ELBOs (higher the better) when the number of samples grows (ND: non-dependent \ref{eq:zi-models-nd}, CD: column-dependent \ref{eq:zi-models-cd}, RD: row-dependent \ref{eq:zi-models-rd}).}
    \end{center}
\end{table*}

\paragraph*{Conclusion of the study}

The Standard method exhibits superior performance across various scenarios,
achieving lower RMSEs and higher ELBO values. However, in regimes with limited
sample sizes, the Enhanced and Enhanced Analytic methods outperformed the
Standard method. This improvement is likely attributable to the regularization
imposed on the parameters, resulting in a more constrained distribution of the
variational parameters.

\section{Application to cow microbiome data}
\label{sec:application}

We now consider a study on the structure and evolution of the microbiota of 45 lactacting cows
before and after calving \citep{Mariadassou2023}. In this experiment, three body
sites in addition to the milk of the 4 teats were sampled (buccal, vaginal, 
nasal)
at 4 times points: 1 week before calving (except for the milk), 1 month, 3 months and 7 months
after calving. The data include $n = 921$ samples with sequencing depths ranging from 1,003 to
81,591 reads. After \citeauthor{Mariadassou2023}'s preprocessing, a total of 1209
Amplicon Sequence Variants (ASV) were identified using the FROGS pipeline \citep{frogs}
based on DADA2 \citep{Callahan2016} and taxonomy was assigned using reference databases.
We filtered out ASV with prevalence lower than 5\% and removed samples for which the total
count was zero, resulting in a count table of $n=880$ samples with $p=259$ 
ASV and a mean proportion of zeroes of 90.3\%.

Previous analyses have already shown that body site and sampling time are strong
structuring factors of the microbiota. We are interested
in comparing the results of PLN and ZIPLN on this dataset, in particular how
the explicit modeling of the Zero-Inflation changes (i) the fit of low value
counts and (ii) the position and clustering of samples in the latent space.

\subsection{Model selection}
We considered a total of 17 models, either without zero inflation (PLN, 4 models)
or with zero inflation (ZIPLN, 13 models). The four PLN
models correspond to different covariate configurations:
intercept-only (no covariates), site only, time
only, and a full model including both site and time
effects. For the ZIPLN models, we considered thirteen
configurations: one model without covariates in either the
PLN or zero-inflation (ZI) components; three models with
covariates only in the ZI component (and none in the PLN
component); and all remaining possible combinations of
covariates between the ZI and PLN components. Based on the
results of \Cref{sec:sim-study}, each ZIPLN model was
fitted with the Standard VA ($\Jun$). We compared model
fits using AIC, BIC, and ICL, as summarized
in~\Cref{tab:model-criterion}.

As expected, BIC tends to
favor more parsimonious models than AIC and ICL. AIC
generally selects denser models across all model types
(PLN, ZI, and ZIPLN), often including both site and time
effects in each component. ICL shows a less consistent
pattern, as it accounts for uncertainty in the latent
variables of both the ZI and PLN components, but it also
tends to prefer denser models. Overall, the model
comparison highlights that accounting for zero inflation
-- particularly the effects of site, time, and their
interaction on it -- has a greater impact than modeling
these covariates in the count component alone, leading to
substantial improvements in penalized likelihood values.
For the subsequent analyses, we focus on the full models:
the PLN model with the site $\times$ time interaction, and
the ZIPLN model with the site $\times$ time interaction in
both components. Although the latter is not selected by
AIC, it is retained for ease of comparison and visualization.

\begin{table*}[h]
\fontsize{10pt}{12.5pt}\selectfont
\begin{tabular*}{\linewidth}{@{\extracolsep{\fill}}rllrrrrr}
\toprule
 & \multicolumn{3}{c}{Model} & \multicolumn{4}{c}{Criterion} \\
\cmidrule(lr){2-4} \cmidrule(lr){5-8}
  & ZI & PLN & \#param & Loglik & BIC & AIC & ICL \\
\midrule\addlinespace[2.5pt]
\multicolumn{8}{l}{PLN} \\[2.5pt]
\midrule\addlinespace[2.5pt]
 &  & \textasciitilde{}1 & 33929 & -219757.4 & {\cellcolor[HTML]{9A2515}{\textcolor[HTML]{FFFFFF}{-334775.4}}} & -253686.4 & -839241.5 \\
 &  & \textasciitilde{} time & 34706 & -217954.9 & -335606.9 & -252660.9 & {\cellcolor[HTML]{9A2515}{\textcolor[HTML]{FFFFFF}{-837857.8}}} \\
 &  & \textasciitilde{} site & 34706 & -217321.5 & -334973.5 & -252027.5 & -842579.4 \\
 &  & \textasciitilde{} site:time & 37555 & -213480.1 & -340790.1 & {\cellcolor[HTML]{9A2515}{\textcolor[HTML]{FFFFFF}{-251035.1}}} & -844789.3 \\
\midrule\addlinespace[2.5pt]
\multicolumn{8}{l}{ZI} \\[2.5pt]
\midrule\addlinespace[2.5pt]
 & \textasciitilde{}1 &  & 34188 & -212293.7 & -328189.6 & -246481.7 & -634481.0 \\
 & \textasciitilde{} time &  & 34965 & -207744.1 & -326274.1 & -242709.1 & -628298.5 \\
 & \textasciitilde{} site &  & 34965 & -203268.7 & {\cellcolor[HTML]{9A2515}{\textcolor[HTML]{FFFFFF}{-321798.7}}} & -238233.7 & -618514.2 \\
 & \textasciitilde{} site:time &  & 37814 & -194612.1 & -322800.1 & {\cellcolor[HTML]{9A2515}{\textcolor[HTML]{FFFFFF}{-232426.1}}} & {\cellcolor[HTML]{9A2515}{\textcolor[HTML]{FFFFFF}{-614484.2}}} \\
\midrule\addlinespace[2.5pt]
\multicolumn{8}{l}{ZIPLN} \\[2.5pt]
\midrule\addlinespace[2.5pt]
 & \textasciitilde{}time & \textasciitilde{}time & 35742 & -207437.7 & -328601.7 & -243179.7 & -624174.2 \\
 & \textasciitilde{}time & \textasciitilde{}site & 35742 & -203564.1 & -324728.1 & -239306.1 & -627213.9 \\
 & \textasciitilde{}site & \textasciitilde{}site & 35742 & -202378.3 & -323542.2 & -238120.3 & -620865.7 \\
 & \textasciitilde{}site & \textasciitilde{}time & 35742 & -201848.1 & {\cellcolor[HTML]{9A2515}{\textcolor[HTML]{FFFFFF}{-323012.1}}} & -237590.1 & -610669.3 \\
 & \textasciitilde{}time & \textasciitilde{}site:time & 38591 & -198613.1 & -329435.1 & -237204.1 & -619692.0 \\
 & \textasciitilde{}site & \textasciitilde{}site:time & 38591 & -196311.1 & -327133.0 & -234902.1 & -611070.3 \\
 & \textasciitilde{} site:time & \textasciitilde{} time & 38591 & -194802.1 & -325624.1 & -233393.1 & -609588.2 \\
 & \textasciitilde{} site:time & \textasciitilde{} site & 38591 & -194090.6 & -324912.6 & {\cellcolor[HTML]{9A2515}{\textcolor[HTML]{FFFFFF}{-232681.6}}} & -614143.3 \\
 & \textasciitilde{} site:time & \textasciitilde{} site:time & 41440 & -191428.4 & -331908.3 & -232868.4 & {\cellcolor[HTML]{9A2515}{\textcolor[HTML]{FFFFFF}{-607483.7}}} \\
\bottomrule
\end{tabular*}
\caption{\label{tab:model-criterion} Model selection criteria (BIC, AIC and ICL) and model details for the 17 ZIPLN models with (ZI and ZIPLN) or without (PLN) zero inflation fitted to the cow microbiota dataset.}
\end{table*}

\subsection{Modeling of counts}

Comparing the model fits of PLN and ZIPLN (\Cref{fig:microcosm_model_fit}) show
that both models are good at fitting the counts along the observed range
(left panel) with a small upward bias for ZIPLN for low counts. However, when
focusing on observed zeroes, we find out that ZIPLN predicts lower values than
PLN (middle panel). The bimodality observed for ZIPLN fitted values corresponds
to a mixture of two populations (right panel): on the one hand, counts for
which both the variational probability of zero-inflation
$P_{ij}$ is close to 1 and the latent mean $M_{ij}$
is very negative and on the other hand counts for which
$M_{ij}$ is much higher but compensated by a high value of $P_{ij}$. The
first population corresponds to zeroes that could be well-captured by a PLN
model without zero inflation whereas the second one corresponds to zeroes which
are only fitted well thanks to the zero inflation component.

\begin{figure*}[btp]
\includegraphics[width=\linewidth]{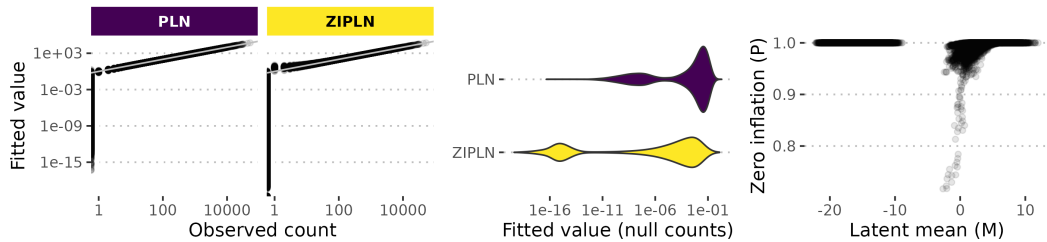}
\caption{\label{fig:microcosm_model_fit} Model fits of PLN and ZIPLN in terms of fitted versus observed counts (left panel), fitted values for null counts (middle panel) and comparaison of $P_{ij}$ and $M_{ij}$ estimated for null counts by ZIPLN (right panel).}
\end{figure*}

Looking at the variational parameters estimated by ZIPLN (\Cref{fig:microcosm_parameters}), we find that the different microbiota are highly structured in blocks of species that are only prevalent and abundant in specific sites at specific times. In particular, few ASVs are detected in the oral microbiota (negative $\matr{M}$ and $\matr{P}$ close to $1$, corresponding to mostly null counts). The blue bands in $\matr{M}$ indicate ASVs likely to be structurally absent from this microbiota. In constrast, the bright yellow blocks in $\matr{P}$, and the corresponding ones in $\matr{M}$, highlight ASVs that are systematically present and abundant in the oral microbiota 1 week after calving ($\matr{P}$ close to $0$ and high values of $\matr{M}$ across all samples in that category). Finally, many ASVs have large positive $\matr{M}$ values in the milk microbiota but a heterogeneous pattern of zero inflation across samples, corresponding to ASVs that are not systematically found in the milk but abundant when present. This is in line with the documented \citep{Mariadassou2023} high diversity and large biological diversity observed in the milk compared to microbiota from other body sites.

\begin{figure*}[btp]
\includegraphics[width=\linewidth]{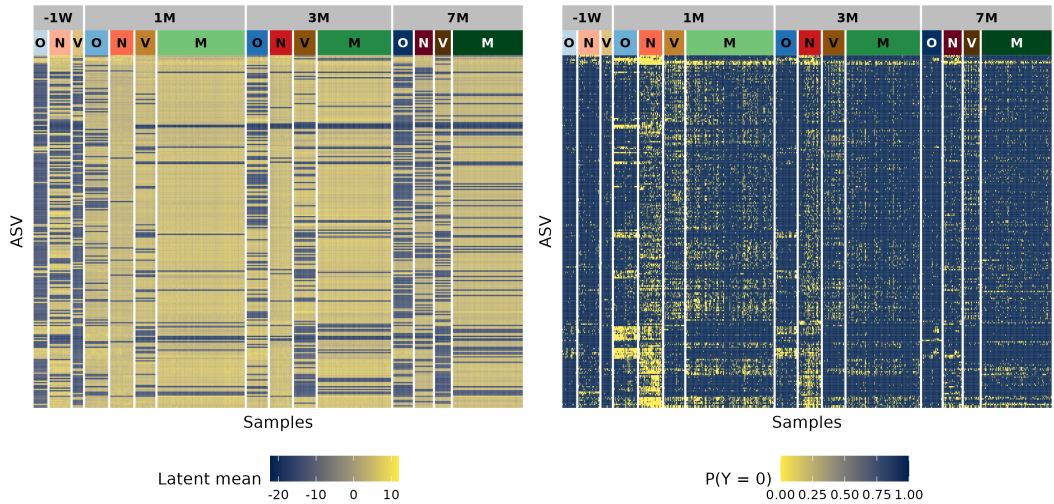}
\caption{\label{fig:microcosm_parameters} Variational latent means $\matr{M}$ (left) and zero inflation probability $\matr{P}$ (right) estimated by ZIPLN. The ASVs are ordered using a hierarchical clustering (Ward linkage) on the centered-log-ratio-transformed observed counts. The first grouping factor is sampling time: 1 week before (-1W), 1 month (1M), 3 months (3M) and 7 months (7M) after calving and the second is the type of microbiota: oral (O), nasal (N), vaginal (V) or from the milk (M).}
\end{figure*}

\subsection{Latent means}

Finally, we focus on the similarity of samples in the latent space. We do so
with a PCA of the latent means $\matr{M}$ (\Cref{fig:microcosm_pca}) inferred
by PLN (left panel) and ZIPLN (right panel). The results of both models are
quite similar with a strong stratification of microbiota according to body site
along the diagonal and according to time along the antidiagonal. Likewise, the
site $\times$ time groups are in the same positions in both panels. The
striking difference between both panels lies in the scale of the within-group
dispersion of the samples: large for PLN and much smaller for ZIPLN, leading to
tighter and better separated groups. This is coherent with the observations
from \Cref{fig:microcosm_model_fit} (right panel): observed zeroes can be
captured by values of $P_{ij} \simeq 1$, without affecting $M_{ij}$, unlike PLN
where observed zeroes systematically lead to highly negative values $M_{ij}$
(mean value of $M_{ij}$ for observed zeroes: $-1.12$ with ZIPLN and $-13.6$
with PLN) therefore translating the samples in the latent space and leading to
higher variability. This is also reflected in the estimate of
$\boldsymbol{\Sigma}$ which has a much smaller volume for ZIPLN ($\log
|\boldsymbol{\Sigma}| = -206.24$, $\text{Tr}(\boldsymbol{\Sigma}) = 216.8$)
than for PLN ($\log |\boldsymbol{\Sigma}| = 689.11$,
$\text{Tr}(\boldsymbol{\Sigma}) = 15 407$) corresponding to narrower
distributions in the latent space.

\begin{figure*}[h]
\includegraphics[width=\linewidth]{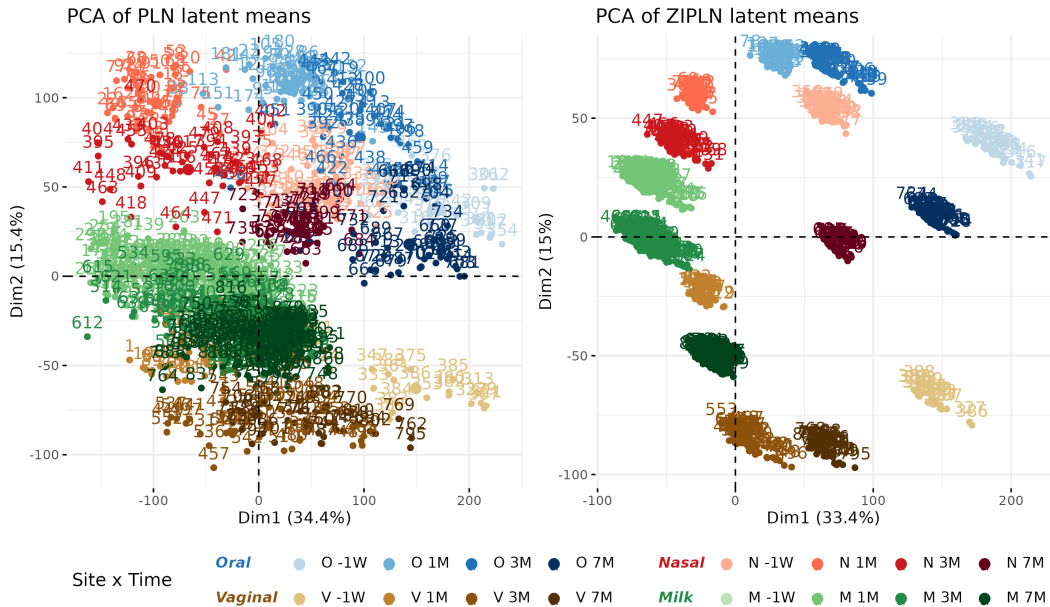}
\caption{\label{fig:microcosm_pca} PCA of the variational latent means $\matr{M}$ inferred by PLN (left) and ZIPLN (right). The samples are colored according to the site $\times$ time categories, using the same abbreviations for site and time as in \Cref{fig:microcosm_parameters}.}
\end{figure*}

\section{Conclusion and perspectives}

\subsection*{Conclusion}
In the context of analysis of high-dimensional count data, we
introduce the ZIPLN model, driven by a latent
Gaussian variable managing the structure between counts and a zero-inflated
component explaining the zeroes that fails to be explained by a standard PLN
model.
The zero-inflation is flexible as it can be fixed, site-specific,
feature-specific or depends on covariatates. We use two variational
approximations, one breaking all dependencies between features and one relying
on conditional law of counts given the observed zeroes. We compare and assess
the quality of both variational approximations on synthetic data and show the
efficiency of ZIPLN even when $90\%$ of the counts are corrupted by
zero-inflation. Our results show that the Standard VA is faster than the Enhanced VA, as expected, but also performs better or almost equivalently in terms of log-likelihood and parameter estimation, a surprise for us.
The model is motivated by an application on the structure of the microbiota of
$45$ lactating cows where $90.3\%$ of counts are null.  We show the
zero-inflation leads to massive improvements in terms of log likelihood and point out that
ZIPLN is indeed capturing zeroes that PLN can not, leading to better and more tightly
separated groups inside the microbiota dataset.\\
\subsection*{Perspectives}
\label{sec:discussion}
While this method easily scales to a couple of thousands of variables $p$ and 
even more when using GPUs, the number of parameters increases quadratically 
with $p$, due to $\matr \Sigma$,
making it difficult to scale to tens of thousands of variables.
A potential extension of this work involves employing a low-rank approximation
of the covariance matrix \citep{Chiquet2017}, enabling scalability to
higher-dimensional datasets. Adapting the PLNPCA model proposed
by \citet{Chiquet2017} to zero-inflated data would require modifying the
Gaussian latent component in \Cref{eq:ZIPLN-model} as follows:
\begin{align*}
    \boldsymbol{F}_i & \sim \mathcal{N}(\matr{0}_q, \matr{I}_q),\\
    \boldsymbol{Z}_i & = \matr{X}_i \matr{B} + \matr{C} \matr{F}_i,
\end{align*}
where $1 \leq q \leq p$ and $\matr{C} \in \mathbb{R}^{p \times q}$ represents a
low-rank loading matrix. This formulation would allow the variational
approximation to focus on the conditional distribution of the low-dimensional
latent variable $\boldsymbol{F}_i$ given the observed counts. Consequently, the
number of model parameters would be significantly reduced, scaling from
$O(p^2)$ to $O(pq)$. An implementation of this model is available in the
\texttt{pyPLNmodels} package.

An alternative approach to reduce the number of model parameters involves
imposing an $\ell_1$ penalty on the inverse covariance matrix, thereby limiting
the number of variables that are strongly correlated, as proposed by
\citet{chiquet2019variational}. Extending the ZIPLN model presented in
\Cref{eq:ZIPLN-model} to this framework requires modifying the latent Gaussian
component as follows:
\begin{align*}
    \boldsymbol{Z}_i & \sim \mathcal{N} \left(\matr{X}_i \matr{B}, \boldsymbol{\Omega}^{-1} \right), \quad \|\boldsymbol{\Omega}\|_1 \leq \lambda,
\end{align*}
where $\lambda > 0$ is a hyperparameter controlling the sparsity of the
precision matrix $\boldsymbol{\Omega}$. The zero-inflation component of the
model remains unchanged. The variational approximation would be adapted to
estimate the sparse precision matrix $\boldsymbol{\Omega}$, incorporating the
$\ell_1$ penalty to enforce sparsity. This approach results in a sparse
covariance structure, which can be estimated using a variational framework
similar to that described by \citet{chiquet2019variational}. An implementation
of this model is available in the \texttt{PLNmodels}
package.

Finally, a second-order Taylor expansion of the function \( x \mapsto \exp(-x)
\), as suggested in \cite{DLEXP}, could be considered to derive an approximate yet
analytic VE-step. This approach is expected to significantly accelerate the
optimization process.

\backmatter

\bmhead{Acknowledgements}

We would like to thank Jean-Benoist Léger for the implementation of the Lambert function and the precious advices to build the \texttt{pyPLNmodels} package.

\bmhead{Funding}

Bastien Bartardière, Julien Chiquet and François Gindraud are supported by
the French ANR grant ANR-18-CE45-0023 Statistics and Machine Learning for Single Cell Genomics (SingleStatOmics).

\bmhead{Code availability}

All the algorithms are available in the python package \texttt{pyPLNmodels} \footnote{\url{https://github.com/PLN-team/pyPLNmodels}} and the R package \texttt{PLNmodels} \footnote{\url{https://github.com/PLN-team/PLNmodels}}.

\begin{appendices}

\section{Technical proofs}
\crefalias{section}{appendix}
\subsection{ELBO (Proposition~\ref{prop:elbo})}

We recall that $\matr{M} = [\matr{M}_1, \dots,
\matr{M}_n]\tr$, $\matr{P} = [\matr{P}_1, \dots,
\matr{P}_n]\tr$, $\matr Q = \matr
1_{n,p} - \matr P$, $\matr{S} = [\matr{S}_1, \dots, \matr{S}_n]\tr$ and
$\bar{\matr{S}}^2 = \matr 1_n \tr \matr S^2$ and $\mEtun$ (resp. $\mEtdeux$) the
expectation under $\ptun$ (resp. $\ptdeux$). 
\review{We now derive the Expected LOwer-Bound for both variational approximations, defined by
$J^{(k)}(\psi,\ta) = \mEtk [\log p_\theta(\matr{Z}, \matr{W}, \matr{Y})] - \mEtk [ \log \pt_\psi^{(k)}(\matr{Z}, \matr{W})],$ where $\mEtk$ stands for the expectation with variational approximation $\pt^{(k)}_{\psi}$ $\left(k = \{1,2\}\right)$.}


We first compute the entropy term $\mEtk \left[\log \pt_{\psi}^{(k)}(\matr Z, \matr W)\right]$ for each variational distribution, starting with the standard one:
\begin{align}
\nonumber - & \mEtun \left[\log \pt_{\psi}(\matr{Z}, \matr{W})\right] \\
    & = \sum_i- \mEtun \left[\log \pt_{\psi_1}(\matr{Z}_i)\right] - \mEtun \left[\log \pt_{\psi_2}(\matr{W}_i)\right] \nonumber\\
 \nonumber & =  \frac12 \sum_i \left( p + {\matr{1}_p\tr\log\matr{S}_i^2} + p\log(2\pi) \right) \\
    & \qquad - \sum_{i,j}  \bigg[P_{ij} \log(P_{ij}) + Q_{ij}\log(Q_{ij})\bigg]. \label{eq:entropy_un}
\end{align}

For the Enhanced approximation, \review{we decompose $\mEtdeux \left[\log \pt_{\psi}\left(\matr{Z}, \matr{W}\right)\right]$ into two terms:} 
\begin{equation}
\mEtdeux \left[\log \pt_{\psi}\left(\matr{W}\right)\right] + \mEtdeux \left[\log \pt_{\psi}\left(\matr{Z}| \matr{W}\right)\right].
    \label{eq:entropy_deux}
\end{equation}
\review{The first term corresponds to the entropy of the Bernoulli distribution, i.e, $\sum_{i,j}  \left[P_{ij} \log(P_{ij}) + Q_{ij}\log(Q_{ij})\right]$}.  For the second term, denoting $ \sumzero \triangleq \sum_{1 \leq i \leq n,1 \leq j \leq p, Y_{ij} = 0}$, $ \sumun \triangleq \sum_{1 \leq i \leq n,1 \leq j \leq p, Y_{ij} >
0}$, $\mathcal N \left(x; \mu, \sigma^2 \right)$ the density of a Gaussian with mean $\mu$ and variance $\sigma^2$ evaluated at $x \in \mathbb R$ and remarking that $P_{ij} = 0$ whenever $Y_{ij} = 0$, we get
\begin{align*}
    \mEtdeux &  \big[\log \pt_{\psi}\left(\matr{Z}| \matr{W}\right)\big] = \\
     & \quad \sumzero\mEtdeux\Big[W_{ij} \log(\mathcal N (Z_{ij}; \matr{x}_i\tr \matr B_j, \Sigma_{jj})) \\
     & \quad + (1 - W_{ij}) \log(\mathcal N (Z_{ij}; M_{ij}, S_{ij}^2)) \Big] \\
   & \quad + \sumun   \mEtdeux \left[\log(\mathcal N (Z_{ij}; M_{ij}, S_{ij}^2))| W_{ij} = 0 \right] \\
   & = - \sumzero P_{ij}\left( \frac{\log(\Sigma_{jj})} 2 - \log(|S_{ij}|)  \right) \\
   & \quad - \sum_{i,j} \left(\log(|S_{ij}|)\right) - \frac{np}{2}\log(2\pi e). \\
   & = - \frac{1}{2}\sum_{ij} P_{ij} \log(\Sigma_{jj}) - \sum_{ij} Q_{ij} \log(|S_{ij}|) \\
   & \quad - \frac{np}{2}\log(2\pi e)
\end{align*}

Now, the complete data log-likelihood of the ZIPLN regression model is given by
\begin{align}\label{eq:log_complete_pln}
    & \log p_\theta(\matr{Z}, \matr{W}, \matr{Y}) = \\
    & \, \log p_\theta(\matr{W}) + \log p_\theta(\matr{Y} | \matr{Z}, \matr{W}) +  \log p_\theta(\matr{Z}) \nonumber \\
    & \, = \sum_{i,j} W_{ij} \matr{x}_i^{0\tr} \matr{B}^0_j - \log(1 + e^{\matr{x}_i^{0\tr} \matr{B}^0_j}) + W_{ij} \delta_{0, \infty}(Y_{ij}) \nonumber \\
    & \, +  (1-W_{ij}) \left( Y_{ij}(o_{ij}  +Z_{ij}) - e^{o_{ij} + Z_{ij}} - \text{cst.} \right) \nonumber  \\
    & \, -\frac{1}{2} \sum_{i=1}^n \left( \|\matr{Z}_i - \matr{x}_i\tr \matr{B}\|_{\matr \Omega}^2 - \logdet{\matr \Omega}  + p \log(2\pi) \right)\nonumber 
\end{align}
We start with computing the expectation $\mEtk \left[\log p_\theta(\matr Z)\right]$, set $\matr{H}_i = \matr{Z}_i -  \matr{x}_i\tr \matr{B}$
and denotes $H_{ij}$ its $j^{\text{th}}$ coordinate. Under the Enhanced approximation we have
 $$H_{ij} | W_{ij} \sim \mathcal N(h_{ij}, s_{ij}^2)^{1 - W_{ij}}\mathcal N(0, \Sigma_{jj})^{W_{ij}}$$
 with $h_{ij} = M_{ij} - \matr{x}_i\tr \matr{B}_j$. Now we have
\begin{equation*}
   \mEtdeux\left[\|\matr{Z}_i - \matr{x}_i\tr \matr{B}\|_{\matr \Omega}^2\right] =\sum_{k=1}^p \sum_{l=1}^p \mEtdeux \left[H_{ik}  \Omega_{kl} H_{il}   \right].
\end{equation*}
Since the $\left(H_{ik}\right)_{1 \leq k \leq p}$ are independent, for all $ 1
\leq k,l \leq p$ we have
$$
\mEtdeux \left[ H_{ik} H_{il} \right] =
\begin{cases}
\mEtdeux\left[ H_{ik} \right]\mEtdeux\left[ H_{il} \right] & \text{if } k \neq l \\
\mEtdeux\left[ H_{ik} \right]^2 + \widetilde{\mathbb V}^{^{(2)}}\left[ H_{ik} \right] & \text{if } k = l,
\end{cases}
$$
where $\widetilde{\mathbb V}^{^{(2)}}$ denotes the variance under the Enhanced approximation.
Furthermore, $\mEtdeux\left[ H_{ij} \right] = \mEtdeux\left[\mEtdeux\left[
H_{ij} | W_{ij} \right]\right] = h_{ij} (1- P_{ij})$. Using the law of total
variance, we also have
\begin{align*}
    & \widetilde{\mathbb V}^{^{(2)}} \left[H_{ij} \right] \\
    & \, = \widetilde{\mathbb V}^{^{(2)}}\left[ \mEtdeux\left[H_{ij} | W_{ij} \right] \right] + \mEtdeux\left[ \widetilde{\mathbb V}^{^{(2)}}\left[H_{ij} | W_{ij} \right] \right]  \\
    & \, = \widetilde{\mathbb V}^{^{(2)}}\left[ h_{ij} (1 - W_{ij}) \right] + \mEtdeux\left[ s_{ij}^2 (1 - W_{ij}) + \Sigma_{jj} W_{ij} \right] \\
    & \,  = h_{ij}^2 P_{ij} Q_{ij} + s_{ij}^2 Q_{ij} + \Sigma_{jj} P_{ij}.
\end{align*}
This gives
\begin{align}
    & \mEtdeux \left[ \| \matr H_i  \|^2_{\Omega} \right] = \sum_{k=1}^p \sum_{l=1}^p \mEtdeux \left[H_{ik}  \Omega_{kl} H_{il}   \right] \nonumber\\
    & \, = \sum_{k=1}^p \sum_{l=1}^p \Omega_{kl} \mEtdeux \left[H_{ik}\right] \mEtdeux \left[H_{il}   \right] + \sum_{j=1}^p \Omega_{jj} \widetilde{\mathbb V}^{^{(2)}} \left[H_{ij}\right]\nonumber \\
    & \, = \mEtdeux \left[ \matr H_i \right]\tr \matr\Omega \mEtdeux \left[ \matr H_i \right]+ \diag(\matr \Omega)\tr \diag(\widetilde{\mathbb V}^{^{(2)}} \left[\matr H_{i}\right]) \nonumber \nonumber  \\
    & \, = \left\|(\matr M_i - \matr x_i \tr \matr B) \odot \matr Q_i \right\|_{\matr \Omega}^2  \nonumber\\
    & \quad + \diag(\matr \Omega)\tr \Big( \matr P_i \odot \matr Q_i \odot (\matr M_i - x_i\tr \matr B )^2 \nonumber \\
    & \qquad + \matr Q_i\odot \matr s_i^2 + \diag(\matr \Sigma) \odot \matr P_i\Big), \label{eq:norm_H_2}
\end{align}
where $\matr Q_i = \matr 1_p - \matr P_i$. For the Standard approximation, a similar argument applies but  $\mEtun\left[ H_{ij}\right] =
M_{ij}$ and $\widetilde{\mathbb V}^{(1)}\left[ H_{ij}\right] =  s_{ij}^2$, giving the much simpler term
\begin{align} \label{eq:norm_H_1}
    \mEtun & \left[\|\matr{Z}_i - \matr{x}_i\tr \matr{B}\|_{\matr \Omega}^2\right] \nonumber \\
    & = \left\| \matr M_i - \matr x_i^{\tr} \matr B \right\|_{\matr \Omega}^2 + \diag(\matr \Omega) \tr \matr s_i^{2}.
\end{align}
As both $\mEtun$ and $\mEtdeux$ coincides on the remaining terms of the complete data
log-likelihood, we drop the index and denote $\tilde{A}_{ij} = \mEt
\left[\exp\left(o_{ij} + Z_{ij}\right)\right] =
\exp\left(o_{ij} + M_{ij} + S^2_{ij}/2\right)$. Taking the variational expectation gives
\begin{align}
& \mEt \big[\log p_\theta\big(\matr Y | \matr{Z},\matr{W} \big) \big] = \sum_{i, j} P_{ij}\delta_{0, \infty}(Y_{ij}) \label{eq:esp_loglike1} \\
  & \, + \sum_{i, j} (1 - P_{ij})(Y_{ij}(o_{ij}+M_{ij}) - \tilde{A}_{ij} - \log(Y_{ij}!)) \nonumber
\end{align}
and
\begin{equation}
\mEt \big[\log \pta(\matr W) \big] =\sum_{i,j} P_{ij} \matr{x}_i^{0\tr} \matr{B}_j^0 - \log(1 + e^{\matr{x}_i^{0\tr} \matr{B}^0_j}) \label{eq:esp_loglike2}
\end{equation}
where all operations ($\exp$, $\log$, $\logit$, etc) are applied component-wise.
Putting \Cref{eq:esp_loglike1,eq:esp_loglike2,eq:entropy_un,eq:norm_H_1}
together gives $\Jun$ and \Cref{eq:esp_loglike1,eq:esp_loglike2,eq:entropy_deux,eq:norm_H_2} gives $\Jdeux$. The writing in compact matrix form is left to the reader.
\subsection{Identifiability (Proposition~\ref{prop:standard-identifiability})}


We use the moments of $\matr{Y}$ to prove identifiability. Letting $A_j = \exp(\mu_j + \sigma_{jj}/2) = \mE[e^{Z_j}]$, and using results on moments of Poisson and Gaussian  distributions,
\begin{itemize}
 \item[(i)] If $U \sim \mathcal{N}(\mu, \sigma^2)$, then $\mE[e^U] = \exp(\mu + \sigma^2/2)$,
 \item[(ii)] If $U \sim \mathcal{P}(\lambda)$ then $\mE[U] = \lambda$, $\mE[U^2] = \lambda (1 + \lambda)$, $\mE[U^3] = \lambda (1 + 3\lambda + \lambda^2)$,
\end{itemize}
we have that $\mE[(e^{Z_j})^2] = A_{j}^2 e^{\sigma_{jj}}$ and $\mE[(e^{Z_j})^3] = A_{j}^3 e^{3\sigma_{jj}}$. By the law of total expectation, we get the first three moments of $Y_j$:
\begin{align*}
 \mE[Y_j] & = (1 - \pi_j) A_j \\
 \mE[Y_j^2] & = (1 - \pi_j) A_j [1 + A_j e^{\sigma_{jj}}] \\
 \mE[Y_j^3] & = (1 - \pi_j) A_j [1 + 3A_j e^{\sigma_{jj}} + A_j^2 e^{3\sigma_{jj}}]
\end{align*}
In parallel, using the conditional independence of $Y_j$ and $Y_k$ ($j \neq k$) knowing $W_j, W_k$, the independence of $W_j, W_k$ and the law of total covariance, we have
$$
\Cov(Y_j, Y_k) = (1 - \pi_j)(1 - \pi_k) A_j A_k (e^{\sigma_{jk}} - 1) .
$$
Using arithmetic manipulations of the moments, we have
\begin{align*}
A_j e^{\sigma_{jj}} & = \frac{\mE[Y_j^2]}{\mE[Y_j]} - 1 = \frac{\mE[Y_j^2] - \mE[Y_j]}{\mE[Y_j]}\\
e^{\sigma_{jj}} & 
= \frac{\mE[Y_j^3] - 3\mE[Y_j^2] + 2\mE[Y_j]}{(\mE[Y_j^2] - \mE[Y_j])^2 / \mE[Y_j]} \\
e^{\mu_j} & = A_j e^{-\sigma_{jj}/2} = A_j e^{\sigma_{jj}}e^{-3\sigma_{jj}/2}\\
& = \frac{(\mE[Y_j^2] - \mE[Y_j])^{4}}{\sqrt{(\mE[Y_j^3] - 3\mE[Y_j^2] + 2\mE[Y_j])^3 \mE[Y_j]}} \\
e^{\sigma_{jk}} & = 1 + \frac{\Cov[Y_j, Y_k]}{\mE[Y_j]\mE[Y_k]} \\
\pi_j & = 1 - \mE[Y_j]/A_j = 1 - \frac{\mE[Y_j]e^{\sigma_{jj}}}{A_je^{\sigma_{jj}}} \\
& = 1 - \frac{\mE[Y_j]^3(\mE[Y_j^3] - 3\mE[Y_j^2] + 2\mE[Y_j])}{(\mE[Y_j^2] - \mE[Y_j])^3}
\end{align*}
Hence, each coordinate of $\matr{\theta}$ can be expressed as a simple
functions of the (first three) moments of the distribution $p_{\matr{\theta}}$
of $Y_j$ and thus $p_{\matr{\theta}} = p_{\matr{\theta}'} \Rightarrow
\matr{\theta} = \matr{\theta}'$.

\begin{proposition}[Derivatives]
    \label{prop:J-gradient-mod-param-1} $\Jun$ has the following first
    order partial derivatives.
    \begin{align*} \label{eq:first-order-derivatives-mod-param-1} \begin{split}
        \frac{\partial \Jun}{\partial \matr{\Omega}} & = \frac{n}{2}
        \matr{\Omega}^{-1} - \frac12 \left[ (\matr{M} -
        \matr{XB})\tr(\matr{M} - \matr{XB}) + \bar{\matr{S}}^2 \right],
        \\
    \frac{\partial  \Jun}{\partial  \matr{B}}  & = \matr{X}\tr\matr{X}
\matr{B} \matr{\Omega}  - \matr{X}\tr \matr{M}\matr{\Omega},\\
\frac{\partial  \Jun}{\partial  \matr{B}^0} & =
                \matr{X}^{0\tr} \matr{P} - \matr{X}^{0\tr} \left(
                \frac{\exp(\matr{X}^0\matr{B}^0)}{1 +
            e^{\matr{X}^0\matr{B}^0}}\right), \\
\frac{\partial \Jun}{\partial \matr{M}} & =
                    (\matr{1}_{n,p} - \matr{P})\odot [\matr{Y} - \matr{A}] -
                    (\matr{M} - \matr{XB}) \matr{\Omega}, \\
                    \frac{\partial  \Jun}{\partial  \matr{S}}  & = \matr{S}^{\oslash}
            - (\matr{1}_{n,p} - \matr{P}) \odot \matr{S} \odot \matr{A} -
        \matr{S} \Diag(\matr{\Omega}), \\
\frac{\partial \Jun}{\partial \matr{\matr P}} & = \matr P \odot (\matr A +  \matr X^0\matr B^0 - \logit(\matr P))    - \log(1- \matr P),
\end{split}
\end{align*}
where $\matr{S}^{\oslash}$ denotes $\frac 1 {\matr S}$ where the division is applied component-wise.
\end{proposition}
\section{Additional simulations}
\label{sec:appendixB}
\paragraph*{Reconstruction error and computation time}
We define the reconstruction error as the RMSE between $\widehat {\matr Y}$ and $\matr Y$ where $\widehat{\matr Y}$ is defined component-wise as
\begin{align*}
    \widehat{Y}_{ij} & = \tilde{\mathbb E} \left[\left(1 - W_{ij}\right)\exp(Z_{ij}) \right] \\
    & =
    \left( 1 - P_{ij}\right)\exp( O_{ij} + M_{ij} + S_{ij}^2 / 2). \\
\end{align*}
We maintain the experimental protocol outlined in Section
\ref{sec:expdetails}, but exclude the Oracle PLN candidate from consideration, 
as its RMSE is computed on a different dataset, making it incomparable to the 
other candidates. Setting $\gamma$ to $\gamma = 2$, following the procedure in 
Section \ref{ssec:pi_fluct}, and fixing the zero-inflation parameter $\matr 
\pi^{\star}$ at $0.3$, as detailed in Section
\ref{ssec:xb_fluct}, we fix the sample size to $n=500$. We then gradually
increase the number of variables $p$ by considering values in
$\Lambda={100,200,\cdots,500}$. For each dimension size in
$\Lambda$, we simulate 30 distinct parameter sets $\theta$, resulting in a
total of $30\times 5$ parameter combinations. Each algorithm is run for 1000 iterations, although convergence is typically achieved within 500
iterations. The results are presented in Figure \ref{fig:dims_computation}.

Compared to the Poisson Log-Normal (PLN) model, both the Non-Analytic VA and
Analytic VA exhibit approximately 2 times and 3 times longer computation times,
respectively. This increased computational demand is attributed to the
additional variational and model parameters required by the zero-inflated
models. The notable disparity between Analytic and Non-Analytic VA stems from
the computation of the Lambert function, which relies on a computationally
intensive fixed-point method. With 500 variables and 1000 iterations, the
computation time averages around 30 seconds, indicating a reasonable
computational burden even for several thousand variables.

Regarding the reconstruction error, all VA methods demonstrate similar performance regardless
of the model choice, revealing a consistent pattern. Notably, the PLN model
yields the best results, contrary to the findings in Figures
\ref{fig:samples_stat}, \ref{fig:poisson_stat}, and \ref{fig:proba_stat}, where
the root mean square error (RMSE) with respect to each model parameter
significantly favored the zero-inflated VA methods. This discrepancy suggests
that accurate parameter estimation is not necessary to achieve good reconstruction error, as has been widely documented in the linear regression setting \citep{raskutti2011}.

\begin{figure*}[h]
        \begin{center}
            \includegraphics[width=\linewidth]{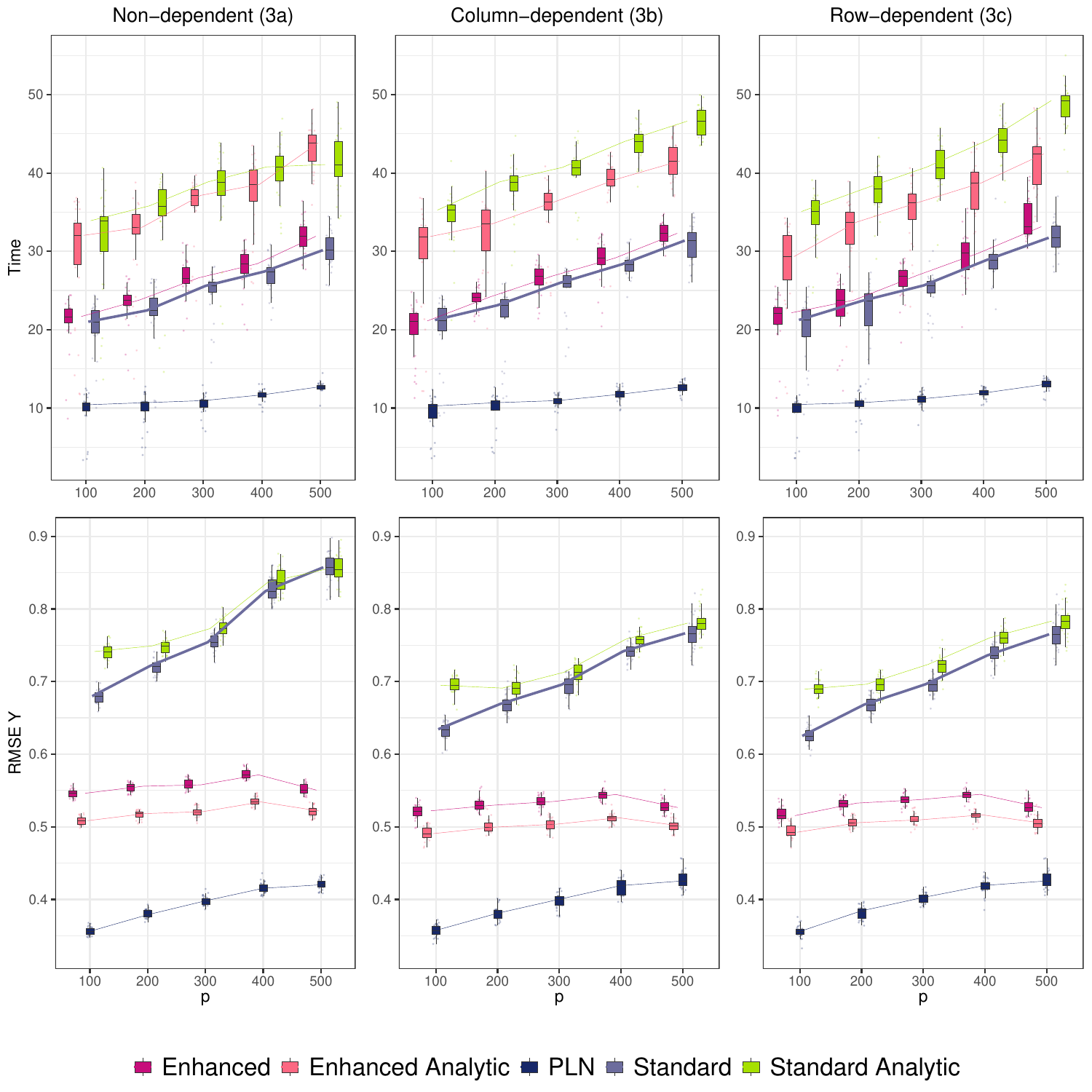}
            \caption{Computation times (in seconds) and reconstruction error (in observation units) when the number of variables $p$ increases. $1000$ iterations were performed for each algorithm.}
            \label{fig:dims_computation}
        \end{center}
\end{figure*}

\end{appendices}

\bibliography{biblio}

\begin{thebibliography}{45}
\providecommand{\natexlab}[1]{#1}
\providecommand{\url}[1]{{#1}}
\providecommand{\urlprefix}{URL }
\providecommand{\doi}[1]{\url{https://doi.org/#1}}
\providecommand{\eprint}[2][]{\url{#2}}
 \bibcommenthead

\bibitem[{Aitchison and Ho(1989)}]{Aitchison1989}
Aitchison J, Ho CH (1989) The multivariate poisson-log normal distribution. Biometrika 76(4):643--653. \doi{10.1093/biomet/76.4.643}

\bibitem[{Akaike(1974)}]{AIC}
Akaike H (1974) A new look at the statistical model identification. IEEE Transactions on Automatic Control 19(6):716--723. \doi{10.1109/TAC.1974.1100705}

\bibitem[{Asmussen et~al.(2014)Asmussen, Jensen, and Rojas-Nandayapa}]{sharpvarphi}
Asmussen S, Jensen JL, Rojas-Nandayapa L (2014) On the laplace transform of the lognormal distribution. Methodology and Computing in Applied Probability 18:441 -- 458. \urlprefix\url{https://api.semanticscholar.org/CorpusID:122716943}

\bibitem[{Batardière et~al.(2024)Batardière, Chiquet, and Mariadassou}]{SandwichBatardiere}
Batardière B, Chiquet J, Mariadassou M (2024) Evaluating parameter uncertainty in the poisson lognormal model with corrected variational estimators. \urlprefix\url{https://arxiv.org/abs/2411.08524}, {\href{https://arxiv.org/abs/2411.08524}{{arXiv:2411.08524}}}

\bibitem[{Biernacki et~al.(2000)Biernacki, Celeux, and Govaert}]{ICL}
Biernacki C, Celeux G, Govaert G (2000) Assessing a mixture model for clustering with the integrated completed likelihood. IEEE transactions on pattern analysis and machine intelligence 22(7):719--725

\bibitem[{Blei et~al.(2017)Blei, Kucukelbir, and McAuliffe}]{blei2017variational}
Blei DM, Kucukelbir A, McAuliffe JD (2017) Variational inference: A review for statisticians. Journal of the American statistical Association 112(518):859--877

\bibitem[{Brault et~al.(2020)Brault, Keribin, and Mariadassou}]{Brault2020}
Brault V, Keribin C, Mariadassou M (2020) Consistency and asymptotic normality of latent blocks model estimators. Electron J Statist 14(1):1234--1268. \doi{10.1214/20-EJS1695}, \urlprefix\url{https://hal.archives-ouvertes.fr/hal-01511960}

\bibitem[{Callahan et~al.(2016)Callahan, McMurdie, Rosen, Han, Johnson, and Holmes}]{Callahan2016}
Callahan BJ, McMurdie PJ, Rosen MJ, et~al (2016) Dada2: High-resolution sample inference from illumina amplicon data. Nature Methods 13(7):581–583. \doi{10.1038/nmeth.3869}, \urlprefix\url{http://dx.doi.org/10.1038/nmeth.3869}

\bibitem[{Cappé et~al.(2002)Cappé, Douc, Moulines, and Robert}]{Cappe_mcmc}
Cappé O, Douc R, Moulines E, et~al (2002) On the convergence of the monte carlo maximum likelihood method for latent variable models. Scandinavian Journal of Statistics 29(4):615--635. \urlprefix\url{http://www.jstor.org/stable/4616738}

\bibitem[{Chiquet et~al.(2017)Chiquet, Mariadassou, and Robin}]{Chiquet2017}
Chiquet J, Mariadassou M, Robin S (2017) Variational inference for probabilistic poisson pca. Ann Appl Stat \urlprefix\url{https://arxiv.org/abs/1703.06633}, {\href{https://arxiv.org/abs/1703.06633}{{1703.06633}}}

\bibitem[{Chiquet et~al.(2019)Chiquet, Robin, and Mariadassou}]{chiquet2019variational}
Chiquet J, Robin S, Mariadassou M (2019) Variational inference for sparse network reconstruction from count data. In: International Conference on Machine Learning, PMLR, pp 1162--1171

\bibitem[{Chiquet et~al.(2021)Chiquet, Mariadassou, and Robin}]{chiquet2021poisson}
Chiquet J, Mariadassou M, Robin S (2021) The poisson-lognormal model as a versatile framework for the joint analysis of species abundances. Frontiers in Ecology and Evolution 9:188

\bibitem[{Cho et~al.(2023)Cho, Liu, Preisser, and Wu}]{BZINB}
Cho H, Liu C, Preisser JS, et~al (2023) A bivariate zero-inflated negative binomial model and its applications to biomedical settings. Statistical Methods in Medical Research 32(7):1300--1317. \doi{10.1177/09622802231172028}, \urlprefix\url{https://doi.org/10.1177/09622802231172028}, pMID: 37167422, {\href{https://arxiv.org/abs/https://doi.org/10.1177/09622802231172028}{{https://doi.org/10.1177/09622802231172028}}}

\bibitem[{Choi et~al.(2022)Choi, Li, and Quon}]{Choi2021.09.15.460498}
Choi Y, Li R, Quon G (2022) sivae: interpretable deep generative models for single-cell transcriptomes. Genome Biology \doi{10.1186/s13059-023-02850-y}, \urlprefix\url{https://genomebiology.biomedcentral.com/articles/10.1186/s13059-023-02850-y}

\bibitem[{Choudhary and Satija(2022)}]{overdispersion}
Choudhary S, Satija R (2022) Comparison and evaluation of statistical error models for scrna-seq. Genome biology 23(1):27. \doi{10.1186/s13059-021-02584-9}, \urlprefix\url{https://europepmc.org/articles/PMC8764781}

\bibitem[{Dempster et~al.(1977)Dempster, Laird, and Rubin}]{DLR77}
Dempster AP, Laird NM, Rubin DB (1977) Maximum likelihood from incomplete data via the {EM} algorithm. Journal of the Royal Statistical Society Series B (Methodological) 39:1--38

\bibitem[{Dong et~al.(2014)Dong, Clarke, Yan, Khattak, and Huang}]{MZINB}
Dong C, Clarke DB, Yan X, et~al (2014) Multivariate random-parameters zero-inflated negative binomial regression model: An application to estimate crash frequencies at intersections. Accident Analysis \& Prevention 70:320--329. \doi{https://doi.org/10.1016/j.aap.2014.04.018}, \urlprefix\url{https://www.sciencedirect.com/science/article/pii/S0001457514001298}

\bibitem[{Escudié et~al.(2017)Escudié, Auer, Bernard, Mariadassou, Cauquil, Vidal, Maman, Hernandez-Raquet, Combes, and Pascal}]{frogs}
Escudié F, Auer L, Bernard M, et~al (2017) {FROGS: Find, Rapidly, OTUs with Galaxy Solution}. Bioinformatics 34(8):1287--1294. \doi{10.1093/bioinformatics/btx791}, \urlprefix\url{https://doi.org/10.1093/bioinformatics/btx791}

\bibitem[{Hui et~al.(2017)Hui, Warton, Ormerod, Haapaniemi, and Taskinen}]{hui2017variational}
Hui FK, Warton DI, Ormerod JT, et~al (2017) Variational approximations for generalized linear latent variable models. Journal of Computational and Graphical Statistics 26(1):35--43

\bibitem[{Jaakkola and Jordan(2000)}]{JaJ00}
Jaakkola TS, Jordan MI (2000) Bayesian parameter estimation via variational methods. Statistics and Computing 10(1):25--37

\bibitem[{Jacquier et~al.(2007)Jacquier, Johannes, and Polson}]{JACQUIER2007615}
Jacquier E, Johannes M, Polson N (2007) Mcmc maximum likelihood for latent state models. Journal of Econometrics 137(2):615--640. \doi{https://doi.org/10.1016/j.jeconom.2005.11.017}, \urlprefix\url{https://www.sciencedirect.com/science/article/pii/S0304407606000704}

\bibitem[{Jin et~al.(2020)Jin, Liu, Li, Xu, Du, Gao, and Xiang}]{Jin_2020}
Jin Y, Liu M, Li Y, et~al (2020) Variational auto-encoder based bayesian poisson tensor factorization for sparse and imbalanced count data. Data Mining and Knowledge Discovery 35(2):505--532. \doi{10.1007/s10618-020-00723-7}, \urlprefix\url{https://doi.org/10.1007%2Fs10618-020-00723-7}

\bibitem[{Kingma and Welling(2022)}]{kingma2022autoencoding}
Kingma DP, Welling M (2022) Auto-encoding variational bayes. {\href{https://arxiv.org/abs/1312.6114}{{arXiv:1312.6114}}}

\bibitem[{Lambert(1992)}]{lambert1992}
Lambert D (1992) Zero-inflated poisson regression, with an application to defects in manufacturing. Technometrics 34(1):1--14

\bibitem[{Lee et~al.(2020)Lee, Coull, Moscicki, Paster, and Starr}]{Lee2020}
Lee KH, Coull BA, Moscicki AB, et~al (2020) Bayesian variable selection for multivariate zero-inflated models: Application to microbiome count data. Biostatistics 21(3):499--517. \doi{10.1093/biostatistics/kxy067}, \urlprefix\url{https://doi.org/10.1093/biostatistics/kxy067}, {\href{https://arxiv.org/abs/https://academic.oup.com/biostatistics/article-pdf/21/3/499/33416271/kxy067.pdf}{{https://academic.oup.com/biostatistics/article-pdf/21/3/499/33416271/kxy067.pdf}}}

\bibitem[{Li(2012)}]{Li2012}
Li CS (2012) {Identifiability of zero-inflated Poisson models}. Brazilian Journal of Probability and Statistics 26(3):306 -- 312. \doi{10.1214/10-BJPS137}, \urlprefix\url{https://doi.org/10.1214/10-BJPS137}

\bibitem[{Li et~al.(1999)Li, Lu, Park, Kim, Brinkley, and Peterson}]{MZIP}
Li CS, Lu JC, Park J, et~al (1999) Multivariate zero-inflated poisson models and their applications. Technometrics 41(1):29--38. \urlprefix\url{http://www.jstor.org/stable/1270992}

\bibitem[{Liu and Zhong(2024)}]{DLEXP}
Liu W, Zhong Q (2024) High-dimensional covariate-augmented overdispersed poisson factor model. Biometrics 80(2):ujae031. \doi{10.1093/biomtc/ujae031}, \urlprefix\url{https://doi.org/10.1093/biomtc/ujae031}, {\href{https://arxiv.org/abs/https://academic.oup.com/biometrics/article-pdf/80/2/ujae031/57351125/ujae031.pdf}{{https://academic.oup.com/biometrics/article-pdf/80/2/ujae031/57351125/ujae031.pdf}}}

\bibitem[{Lopez et~al.(2018)Lopez, Regier, Cole, Jordan, and Yosef}]{Lopez2018DeepGM}
Lopez R, Regier J, Cole M, et~al (2018) Deep generative modeling for single-cell transcriptomics. Nature methods 15:1053 -- 1058. \urlprefix\url{https://api.semanticscholar.org/CorpusID:53643161}

\bibitem[{Love et~al.(2014)Love, Huber, and Anders}]{DEseq2}
Love MI, Huber W, Anders S (2014) Moderated estimation of fold change and dispersion for rna-seq data with deseq2. Genome biology 15(12):1--21

\bibitem[{Mariadassou et~al.(2023)Mariadassou, Nouvel, Constant, Morgavi, Rault, Barbey, Helloin, Rué, Schbath, Launay, Sandra, Lefebvre, Le~Loir, Germon, Citti, and Even}]{Mariadassou2023}
Mariadassou M, Nouvel LX, Constant F, et~al (2023) Microbiota members from body sites of dairy cows are largely shared within individual hosts throughout lactation but sharing is limited in the herd. Animal Microbiome 5(1):1--17. \doi{10.1186/s42523-023-00252-w}, \urlprefix\url{https://doi.org/10.1186/s42523-023-00252-w}

\bibitem[{Niku et~al.(2019)Niku, Hui, Taskinen, and Warton}]{gllvm}
Niku J, Hui FK, Taskinen S, et~al (2019) gllvm: Fast analysis of multivariate abundance data with generalized linear latent variable models in r. Methods in Ecology and Evolution 10(12):2173--2182

\bibitem[{O'Hara and Kotze(2010)}]{NoLogTransform}
O'Hara R, Kotze D (2010) Do not log-transform count data. Methods in Ecology and Evolution 1:118--122

\bibitem[{Raskutti et~al.(2011)Raskutti, Wainwright, and Yu}]{raskutti2011}
Raskutti G, Wainwright MJ, Yu B (2011) Minimax rates of estimation for high-dimensional linear regression over $\ell_q$ -balls. IEEE Transactions on Information Theory 57(10):6976--6994. \doi{10.1109/TIT.2011.2165799}

\bibitem[{Risso et~al.(2018)Risso, Perraudeau, Gribkova, Dudoit, and Vert}]{risso2018}
Risso D, Perraudeau F, Gribkova S, et~al (2018) A general and flexible method for signal extraction from single-cell rna-seq data. Nature communications 9(1):1--17

\bibitem[{Robbins and Monro(1951)}]{SGD}
Robbins H, Monro S (1951) A stochastic approximation method. The annals of mathematical statistics pp 400--407

\bibitem[{Rojas-Nandayapa(2008)}]{lognormcharact}
Rojas-Nandayapa L (2008) Risk probabilities: asymptotics and simulation. PhD thesis, Aarhus Universitetsforlag

\bibitem[{Schwarz(1978)}]{BIC}
Schwarz G (1978) {Estimating the Dimension of a Model}. The Annals of Statistics 6(2):461 -- 464. \doi{10.1214/aos/1176344136}, \urlprefix\url{https://doi.org/10.1214/aos/1176344136}

\bibitem[{Seabold and Perktold(2010)}]{statsmodels}
Seabold S, Perktold J (2010) statsmodels: Econometric and statistical modeling with python. In: 9th Python in Science Conference

\bibitem[{Stoehr and Robin(2024)}]{stoehr2024composite}
Stoehr J, Robin SS (2024) Composite likelihood inference for the poisson log-normal model. {\href{https://arxiv.org/abs/2402.14390}{{arXiv:2402.14390}}}

\bibitem[{Wainwright and Jordan(2008)}]{WAJ08}
Wainwright MJ, Jordan MI (2008) Graphical models, exponential families, and variational inference. Found Trends Mach Learn 1(1--2):1--305. \urlprefix\url{http:/dx.doi.org/10.1561/2200000001}

\bibitem[{Wang and Gu(2018)}]{WANG2018320}
Wang D, Gu J (2018) Vasc: Dimension reduction and visualization of single-cell rna-seq data by deep variational autoencoder. Genomics, Proteomics \& Bioinformatics 16(5):320--331. \doi{https://doi.org/10.1016/j.gpb.2018.08.003}, \urlprefix\url{https://www.sciencedirect.com/science/article/pii/S167202291830439X}, bioinformatics Commons (II)

\bibitem[{Westling and McCormick(2019)}]{sandwich}
Westling T, McCormick TH (2019) Beyond prediction: A framework for inference with variational approximations in mixture models. Journal of Computational and Graphical Statistics 28(4):778--789. \doi{10.1080/10618600.2019.1609977}, \urlprefix\url{https://doi.org/10.1080/10618600.2019.1609977}, pMID: 32713999, {\href{https://arxiv.org/abs/https://doi.org/10.1080/10618600.2019.1609977}{{https://doi.org/10.1080/10618600.2019.1609977}}}

\bibitem[{Xu et~al.(2023)Xu, Xu, Meng, Lu, Cai, Zeng, Nussinov, and Cheng}]{XU2023100382}
Xu J, Xu J, Meng Y, et~al (2023) Graph embedding and gaussian mixture variational autoencoder network for end-to-end analysis of single-cell rna sequencing data. Cell Reports Methods 3(1):100382. \doi{https://doi.org/10.1016/j.crmeth.2022.100382}, \urlprefix\url{https://www.sciencedirect.com/science/article/pii/S2667237522002879}

\bibitem[{Zhao et~al.(2020)Zhao, Rai, Du, Buntine, Phung, and Zhou}]{VAE_count}
Zhao H, Rai P, Du L, et~al (2020) Variational autoencoders for sparse and overdispersed discrete data. In: Chiappa S, Calandra R (eds) Proceedings of the Twenty Third International Conference on Artificial Intelligence and Statistics, Proceedings of Machine Learning Research, vol 108. PMLR, pp 1684--1694, \urlprefix\url{https://proceedings.mlr.press/v108/zhao20c.html}

\end{thebibliography}

\end{document}